\documentclass[sigplan,screen,natbib,nocopyright,nonacm]{acmart}

\setcopyright{none}
\copyrightyear{2018}
\acmYear{2018}
\acmDOI{10.1145/1122445.1122456}

\acmConference[Woodstock '18]{Woodstock '18: ACM Symposium on Neural
  Gaze Detection}{June 03--05, 2018}{Woodstock, NY}
\acmBooktitle{Woodstock '18: ACM Symposium on Neural Gaze Detection,
  June 03--05, 2018, Woodstock, NY}
\acmPrice{15.00}
\acmISBN{978-1-4503-XXXX-X/18/06}

\usepackage{amsmath,stmaryrd,color}

\usepackage{cmll}

\usepackage[all,2cell]{xy}
\UseAllTwocells
\SilentMatrices

\newif\iflong\longtrue
\newif\ifwithappendix\withappendixtrue

\usepackage{amsthm}
\theoremstyle{plain}
\newtheorem{theorem}{Theorem}
\newtheorem{lemma}[theorem]{Lemma}

\newtheorem{proposition}[theorem]{Proposition}
\newtheorem*{proposition*}{Proposition}
\newtheorem*{claim*}{Claim}

\theoremstyle{definition}
\newtheorem{definition}[theorem]{Definition}
\newtheorem{example}[theorem]{Example}
\newtheorem*{problem*}{Problem}
\newtheorem*{question*}{Question}

\theoremstyle{remark}
\newtheorem{remark}[theorem]{Remark}
\newtheorem*{notation}{Notation}

\usepackage{bcprules}
\usepackage{symbols}

\begin{document}

\title[Linear-Algebraic Models of Linear Logic as Categories of Modules over $\Sigma$-Semirings]{Linear-Algebraic Models of Linear Logic \\ as Categories of Modules over $\Sigma$-Semirings}

\author{Takeshi Tsukada}
\orcid{0000-0002-2824-8708}
\affiliation{
  \institution{Chiba University}
  \country{Japan}}
\email{tsukada@math.s.chiba-u.ac.jp}
\author{Kazuyuki Asada}
\orcid{0000-0001-8782-2119}
\affiliation{
  \institution{Tohoku University}
  \country{Japan}}
\email{asada@riec.tohoku.ac.jp}

\begin{abstract}
  A number of models of linear logic are based on or closely related to linear algebra, in the sense that morphisms are ``matrices'' over appropriate coefficient sets.
  Examples include models based on coherence spaces, finiteness spaces and probabilistic coherence spaces, as well as the relational and weighted relational models.
  This paper introduces a unified framework based on module theory, making the linear algebraic aspect of the above models more explicit.
  Specifically we consider modules over \emph{\( \Sigma \)-semirings} \( R \), which are ring-like structures with partially-defined countable sums, and show that morphisms in the above models are actually \( R \)-linear maps in the standard algebraic sense for appropriate \( R \).
  An advantage of our algebraic treatment is that the category of \( R \)-modules is locally presentable, from which it easily follows that this category becomes a model of intuitionistic linear logic with the cofree exponential.
  We then discuss constructions of classical models and show that the above-mentioned models are examples of our constructions.
\end{abstract}



\maketitle

\section{Introduction}\label{sec:intro}
A function \( f \) is \emph{linear} if
\begin{equation}
    f(x + y) ~=~ f(x) + f(y)
    \quad\mbox{and}\quad
    f(\alpha \cdot x) ~=~ \alpha \cdot f(x)
    \label{eq:intro:linearity}
\end{equation}
for every coefficient \( \alpha \) from an appropriate set.
Linearity is an important concept that appears frequently in the mathematical sciences.
In the context of mathematical logic, the linearity has been discovered by Girard in his influential paper~\cite{Girard1987}, which proposed \emph{linear logic}.

A number of important models of linear logic are closely related to linear algebra.
Examples include Girard's original model, known as the \emph{coherence space model}~\cite{Girard1987}, the \emph{relational model}, the \emph{finiteness space model}~\cite{Ehrhard2005a}, the \emph{probabilistic coherence space model}~\cite{Girard2004,Danos2011} and the \emph{weighted relational models}~\cite{Laird2013,Laird2016}.
In these models, an object has a specified ``basis'' and a morphism can be represented as a ``matrix''.

The goal of this paper is to provide a unified account of these models from the view point of \emph{abstract linear algebra}.
Given a field \( K \), an \emph{abstract} vector space \( V \) over \( K \) is an additive group together with an action \( k \cdot x \) of \( k \in K \) to \( x \in V \) and a \emph{\( K \)-linear map} \( f \colon V \longrightarrow W \) is a function that satisfies Equation~\ref{eq:intro:linearity}.
Many notions such as the dual space and the tensor product can be defined in this abstract setting and, when one chooses bases, the composition of linear maps can be related to the matrix calculus.
This paper realises the above-mentioned models of linear logic in some categories of abstract linear algebras, by introducing appropriate algebraic structures to their ``coefficients'', namely \emph{\( \Sigma \)-semirings}.

\subsection{Matrix models}
As mentioned above, in many models of linear logic, a morphism can be seen as a matrix.
The simplest example is the relational model \( \CRel \), in which an object is a set \( X \) and a morphism \( f \colon X \longrightarrow Y \) is a relation \( f \subseteq X \times Y \).
By identifying a relation \( f \subseteq X \times Y \) with its characteristic function \( X \times Y \longrightarrow \{0,1\} \), a morphism \( f \colon X \longrightarrow Y \) in \( \CRel \) can be seen as an \( (X \times Y) \)-matrix \( f = (f_{x,y})_{x \in X, y \in X} \) where \( f_{x,y} = 1 \) if \( (x,y) \in f \) and \( f_{x,y} = 0 \) if \( (x,y) \notin f \).
The composition of \( f \colon X \longrightarrow Y \) and \( g \colon Y \longrightarrow Z \) in \( \CRel \) coincides with the matrix composition in which the \( (x,z) \)-entry is given by \( \bigvee_{y \in Y} g_{y,z} \cdot f_{x,y} \).
Hence \( \CRel \) is the category of matrices over the boolean semiring \( B = \{ 0,1 \} \), where the sum is the join.

A characteristic feature of models of linear logic is that we need to deal with matrices of infinite dimensions, because the exponential space \( !X \) of a space \( X \) is often countably infinite dimensional even if \( X \) is finite dimensional.
Infinite dimensional matrices are hard to deal with since the composition involves infinite sums which do not necessarily converge.
This subtlety of infinite dimensional matrices is a motivation of Ehrhard's work on K\"othe sequence spaces~\cite{Ehrhard2002a} and finiteness spaces~\cite{Ehrhard2005a}.

A simple but powerful idea to overcome the problem of infinite sums is to use a ring with (totally-defined) infinite sums as a coefficient ring.
This idea leads to a family of \emph{weighted relational models}~\cite{Laird2013,Laird2016} parametrised by semirings \( \srig \) with infinite sums.\footnote{The requirements for the semiring \( \srig \) slightly varies in the papers.}
This is one of the most general frameworks to construct ``matrix models'' of linear logic.

However some important models, including the coherence space model~\cite{Girard1987}, the finiteness space model~\cite{Ehrhard2005a} and the probabilistic coherence space model~\cite{Girard2004,Danos2011}, are not covered by the theory of weighted relational models.
A standard approach to construct these models is based on \emph{gluing and orthogonality} studied by Hyland and Schalk~\cite{Hyland2003}.
Intuitively (a special case of) their construction considers as an object a triple \( (A, X, U) \) of a space \( A \), a subspace \( X \subseteq A \) and a subspace \( U \subseteq (X \multimap \bot) \) of the dual space of \( X \) that satisfies a certain condition; a morphism from \( (A, X, U) \) to \( (B, Y, V) \) is a morphism \( f \colon A \longrightarrow B \) that maps \( x \in X \) to \( f(x) \in Y \) and \( v \in V \) to \( v \circ f \in U \). 
Although this general construction explains many aspects of above mentioned models, there is still a subtlety (see Section~\ref{sec:gluing} and \cite[Example~66(3)]{Hyland2003}).

\subsection{Coherence spaces and vector spaces}
This paper develops an algebraic theory of models of linear logic based on Equation~\ref{eq:intro:linearity}.
To this end, we need to find an appropriate algebraic structure on objects of above models characterising morphisms of the models via Equation~\ref{eq:intro:linearity}.

Our starting point is the following classical characterisation of linear maps between coherence spaces (see Section~\ref{sec:monoid:example:coherence} for the definitions of notions and symbols).
\begin{proposition}[see, e.g., Definition~10 in \cite{Girard1995a}]\label{prop:intro:coh-linear}
    Let \( X \) and \( Y \) be coherence spaces.
    Then a set-theoretic function \( f \colon \Config(X) \longrightarrow \Config(Y) \) between cliques is a linear map if and only if
    it preserves the disjoint union, i.e., for every family \( (u_i)_{i \in I} \) of pairwise-disjoint cliques \( u_i \in \Config(X) \),
    \begin{equation}
        \textstyle
        f(\biguplus_{i \in I} u_i) \quad=\quad \biguplus_{i \in I} f(u_i).
        \nonumber
\end{equation}
    Here \( \biguplus \) is the disjoint union.
\qed
\end{proposition}

This proposition suggests that a coherence space \( X \) should be regarded as an algebra \( X = (\Config(X), \uplus) \) equipped with the disjoint union \( \uplus \).
This operation satisfies a number of desired lows (e.g.~the existence of the unit \( \emptyset \), the commutativity \( u \uplus u' = u' \uplus u \) and the associativity \( (u \uplus u') \uplus u'' = u \uplus (u' \uplus u'') \)), but it differs from operations of standard algebras in that the disjoint union \( \uplus \) is only partially defined even as a binary operation, i.e.~\( u \uplus u' \) may be undefined for some \( u, u' \in \Config(X) \).

Now the question is as follows: If a coherence space \( X = (\Config(X), \uplus) \) as an algebra equipped with disjoint union \( \uplus \) would be a ``vector space'', what should be the algebra of its coefficients?
The underlying set should be \( \{ 0,1 \} \) because a clique \( u \in \Config(X) \) is a subset \( u \subseteq |X| \) of a set \( |X| \), which can be identified with \( u \in \{0,1\}^{|X|} \).
The sum must be partial; otherwise the induced sum on \( \{0,1\}^{|X|} \) is total and differs from the disjoint union \( \uplus \).
Let us consider the sum on \( \{0,1\} \) defined by
\begin{gather}
    0 + 0 = 0
    \quad\;\;
    0 + 1 = 1 + 0 = 1
    \quad\;\;
    (1 + 1) \mbox{ is undefined}
    \label{eq:intro:coh-sum}
\end{gather}
and let \( \CohRing \) be the algebra with the underlying set \( \{0,1\} \) and the above sum.
Then the set \( \CohRing^{|X|} \) equipped with the coordinate-wise (or point-wise) sum is isomorphic to \( (\mathcal{P}(|X|), \uplus) \).
In particular, if \( u,u' \in \CohRing^{|X|} \) has a non-empty intersection, i.e.~\( u(x) = 1 = u'(x) \) for some \( x \in |X| \), then the sum \( u + u' \) is undefined since the \( x \)-coordinate \( u(x) + u'(x) \) is undefined.

This paper develops this idea in more details and shows that the category of coherence spaces is indeed the category of vector spaces over \( \CohRing \) in a certain sense.

\subsection{$\Sigma$-semirings, $\Sigma$-modules and linear maps}
We generalise the above observations, giving axioms for the above-mentioned algebraic structure and developing the module theory over it.

Recall that the algebraic structure that we seek should have the following features:
(1) It has infinite sums in order to deal with infinite dimensional matrices, and (2)
even finite sums may be undefined.
Several algebras with this kind of additive structure can be found in the literature~\cite{Haghverdi2001,Hines2007,Manes1986}, among which we use \emph{\( \Sigma \)-monoids}~\cite{Haghverdi2000,Hoshino2012b,Tsukada2018} in this paper after technical considerations.

Once the additive structure is established, an appropriate notion of corresponding ``commutative rings'' can be canonically defined as a commutative monoid in the category of the additive algebras.
We call this ring-like algebra a \emph{$\Sigma$-semiring}.
It has partial countable sums and \emph{total} binary products.

Examples of \( \Sigma \)-semirings contain the following.
\begin{itemize}
    \item \( \CohRing \) has the underlying set \( \{ 0,1 \} \), the sum defined in Equation~\eqref{eq:intro:coh-sum}, and the standard product of integers.
\item \( \FinRing \) has the underlying set \( \{ 0,1 \} \) and the standard product of integers.
    Every finite sum \( 1 + 1 + \cdots + 1 \) of \( 1 \) converges to \( 1 \) but the countable sum \( 1 + 1 + \cdots \) of \( 1 \) is undefined.
    \item \( [0,1] \) has the underlying set \( \{ x \in \mathbb{R} \mid 0 \le x \le 1 \} \).
    The sum and product are that of real numbers, except that the sum is undefined if \( \sum_i r_i = \infty \) or \( \sum_i r_i > 1 \).
\end{itemize}
\( \Sigma \)-semirings \( \CohRing \),
\( \FinRing \) and \( [0,1] \) are coefficient semirings for coherence spaces,
finiteness spaces~\cite{Ehrhard2005a} and probabilistic coherence spaces~\cite{Girard2004,Danos2011}, respectively.

Given a \( \Sigma \)-semiring \( \srig \), the category \( \SMod[\srig] \) of \( \srig \)-modules and \( \srig \)-linear maps can be defined in a straightforward way, using (a partial and infinitely variant of) Equation~\ref{eq:intro:linearity}.
We prove that \( \SMod[\srig] \) is a model of intuitionistic linear logic with the cofree exponential modality.
The proof of the existence of the cofree exponential is surprisingly simple, thanks to our algebraic treatment: we just syntactically examine the axiom for \( \srig \)-modules, ensuring that it belongs to the class known as \emph{partial Horn theory}~\cite{Palmgren2007}.
Then it follows that the category \( \SMod[\srig] \) is \emph{locally presentable}~(cf.~\cite{Adamek1994}), and every category that is symmetric monoidal closed and locally presentable is known to have the cofree exponetial~\cite{Barr1991,Porst2008}.

\subsection{Constructing classical models}
The category \( \SMod[\srig] \) of \( \srig \)-modules is a model of intuitionistic logic and not a model of classical logic, i.e.~the double negation \( \smod^{\bot\bot} \defe ((\smod \multimap \srig) \multimap \srig) \) of an \( \srig \)-module \( \smod \in \SMod[\srig] \) is not necessarily isomorphic to \( \smod \).\footnote{Actually, since \( \SMod[\srig] \) is locally presentable, it cannot be a model of classical logic: locally presentable categories are not self-dual except for trivial counterexamples~\cite{Adamek1994}.}

A standard approach to construct a model of classical logic is to restrict the dimension.
For example, the category \( \SVec[\mathbb{R}] \) of real vector spaces is not a classical model but its subcategory of finite dimensional vector spaces is.
The corresponding notion in our setting should be modules with countable bases since modules in our setting have countable sums.
We introduce two notions of basis of an \( \srig \)-module and show that the subcategories of countable dimensional \( \srig \)-modules are models of classical logic.
Many known models of classical linear logic, such as models based on countable coherence spaces, countable finiteness spaces and
countable probabilistic coherence spaces, are equivalent to the countable-dimensional \( \srig \)-module models for appropriate \( \srig \).

\subsection{Related work}
Some closely related work has already discussed above.

A pioneering work on the infinite algebra in the context of linear logic is Lamarche~\cite{Lamarche1992}.
Given a ring-like algebra \( R \) with totally-defined infinite sum, Lamarche studied the free \( R \)-algebra generated by a given set \( X \), which is the algebra of formal power series over \( X \), and the opposite of the category of free \( R \)-algebras, which is cartesian closed.

Algebras with partial infinite sums have studied in \cite{Haghverdi2000} in the context of \emph{geometry of interaction}~\cite{Girard1989a}, in an attempt to develop a mathematical background of the \emph{execution formula}.
Although both \cite{Haghverdi2000} and our paper studies linear logic using partial algebras, the approaches seem fairly different and the connection has not been understood well.

In the context of program semantics, a technique to approximate a program as a formal countable sum of simpler expressions, called the \emph{Taylor expansion}~\cite{Ehrhard2008}, has been an attractive topic of research.
The method of relating formal sums to semantic sums has proved to be useful in the analysis of programs~\cite{Laird2013}, in particular in the construction of a fully abstract model of probabilistic programs~\cite{Danos2011,Ehrhard2014}, and there is an attempt~\cite{Tsukada2017,Tsukada2018} to generalise this approach.
We expect that our model would have a sufficient structure to interpret Taylor expansions of programs, but much remains to be done.

As for the construction of exponentials, the approach used in this paper based on locally presentable categories~\cite{Barr1991,Porst2008} does not seem popular in the community.
A major approach is based on countable biproducts, used in \cite{Laird2013,Laird2016,Tsukada2018} for example.
In the absence of countable biproducts, a standard remedy was given by Melli{\`e}s et al.~\cite{Mellies2009c,Mellies2018a}.
It is used to prove that the exponential modality of probabilistic coherence spaces given in \cite{Danos2011} is actually free~\cite{Abou-Saleh2013}.
The relationship of our construction to \cite{Mellies2009c,Mellies2018a} is currently unclear.
We think that their approach cannot be applicable to general modules, since their approach requires the preservation of certain limits by the tensor product.
This condition can be seen as a kind of flatness, which general modules does not usually have.
It might be the case that modules with bases studied in Section~\ref{sec:classic} satisfies the criteria of \cite{Mellies2009c,Mellies2018a}.

 \section{Partial Algebra Models of Linear Logic}\label{sec:partial}
This section presents a general theory of partial algebras of infinite arity and their relevance to linear logic.
To completely understand the technical materials in this section, some familiarity with the theory of \emph{locally presentable categories} (see a standard textbook~\cite{Adamek1994}) would be required.
However, once the reader accepts some of the results (including Theorems~\ref{thm:pre:horn-presentable} and \ref{thm:pre:presentable-exponential}), the rest of the paper should be understandable without the knowledge of the theory.

Let \( \kappa \) be a regular cardinal; this paper mainly considers the case that \( \kappa = \aleph_1 \).

\begin{notation}
    For expressions \( e \) and \( e' \) possibly containing partial operations, \( e = e' \) means that \emph{both \( e \) and \( e' \) are defined} and their values coincide.
    This interpretation allows us to describe the following predicates using only the equality and logical symbols:
    \begin{itemize}
        \item \( \IsDefined{e} \defp (e = e) \) meaning that \( e \) is defined;
        \item \( (e \Kle e') \defp (\IsDefined{e} \Longrightarrow e = e') \) meaning that, if \( e \) is defined, so is \( e' \) and their values coincide; and
        \item \( (e \Keq e') \defp ((e \Kle e') \wedge (e' \Kle e)) \) meaning that, if one side is defined, so is the other side and their values coincide.
    \end{itemize}
    \( \IsUndef{e} \) is defined as \( \neg (\IsDef{e}) \), meaning that \( e \) is undefined.
    We sometimes use a special expression \( \Undef \), which is undefined.
    Then \( \IsUndef{e} \) is equivalent to \( e \Keq \Undef \).
\end{notation}

\subsection{Partial Algebras}
Let \( S \) be a set of \emph{basic sorts}.
An \emph{\( S \)-sorted signature of \( \kappa \)-ary algebra} is a set \( \Sigma \) of \emph{operation symbols} together with a function that assigns to each operation symbol \( \sigma \) a sort of the form \( \prod_{i < \alpha} s_i \longrightarrow s \) where \( \alpha < \kappa \) and \( s_i, s \in S \).
Let \( X \) be an \( S \)-sorted set \( (X_s)_{s \in S} \) of variables, which we usually write as \( \{ x \colon s \mid s \in S, x \in X_s \} \).
The \( S \)-sorted set \( (\Term^s_{\Sigma}(X))_{s \in S} \) of \emph{terms} over \( X \) is defined inductively as follows: if \( x \in X_s \), then \( x \in \Term^s_\Sigma(X) \); if \( \sigma \) is an operation symbol of sort \( \prod_{i < \alpha} s_i \longrightarrow s \) and \( (t_i)_{i < \alpha} \) is a family of terms such that \( t_i \in \Term^{s_i}_\Sigma(X) \), then \( \sigma(t_i)_{i < \alpha} \in \Term^{s}_\Sigma(X) \).
If \( \sigma \) is nullary, we shall often abbreviate \( \sigma() \) as \( \sigma \).

A \emph{partial algebra} \( A \) of the \( S \)-sorted signature \( \Sigma \) is an \( S \)-sorted set \( (A_s)_{s \in S} \) together with an assignment to each operation \( \sigma \) of sort \( \prod_{i < \alpha} s_i \longrightarrow s \) a partial function \( \sigma_A \colon \prod_{i < \alpha} A_{s_i} \rightharpoonup A_s \).
For an \( S \)-sorted set \( X = (X_s)_{s \in S} \) of variables, a \emph{valuation} over \( X \) is a family \( \varrho = (\varrho_s)_{s \in S} \) of functions \( \rho_s \colon X_s \longrightarrow A_s \).
A \emph{homomorphism} from a partial algebra \( A \) to a partial algebra \( B \) is a family \( h = (h_s)_{s \in S} \) of total functions \( h_s \colon A_s \longrightarrow B_s \) that preserves operations in the following sense: for every \( \sigma \) of sort \( \prod_{i < \alpha} s_i \longrightarrow s \),
\begin{equation*}
    h_s(\sigma_A(a_i)_{i < \alpha}) \qquad\Kle\qquad \sigma_B(h_{s_i}(a_i))_{i < \alpha}.
\end{equation*}
\( \PAlg(\Sigma) \) denotes the category of partial algebras of the signature \( \Sigma \) and homomorphisms.

\subsection{Horn Theories}
Let \( \Sigma \) be an \( S \)-sorted signature of \( \kappa \)-ary algebra.
We introduce a class of formulas describing properties of partial algebras of the signature \( \Sigma \).

A \emph{Horn formula} of arity \( \kappa \) is a formula of the form
\begin{equation*}
    \forall (x_i \colon s_i)_{i < \alpha}.
    \qquad
    \bigwedge_{j < \beta} t_j = t'_j
    \quad\Longrightarrow\quad
    t = t'
\end{equation*}
where \( \alpha,\beta < \kappa \), \( (x_i)_\alpha \) is a family of pairwise-distinct variable, \( s_i \in S \) and \( t_j, t_j', t, t' \in \bigcup_{s \in S} \Term^s_\Sigma(\{ x_i \colon s_i \mid i < \alpha \}) \).
A \emph{Horn theory} is a collection of Horn formulas.

The meaning of the formula should be obvious: let us recall that \( e = e' \) for partially defined expressions \( e, e' \) means that \emph{both \( e \) and \( e' \) are defined} and their values coincide.
Given a partial algebra \( A \) and a Horn theory \( \Th \), we write \( A \models \Th \) to mean that \( A \) satisfies all formulas in \( \Th \).
The formal definition of the judgement \( A \models \Th \), as well as a judgement \( A, \varrho \models \varphi \) for a formula with free variables, should be obvious and omitted.

The category of models of a Horn theory, i.e.~the full subcategory \( \PAlg(\Sigma,\Th) \) of \( \PAlg(\Sigma) \) consisting of algebras \( A \) such that \( A \models \Th \), is a locally presentable category.
This fact has been proved in \cite{Palmgren2007} for the finitely case, and this paper uses a slight generalisation, namely the result for \( \kappa = \aleph_1 \).\footnote{\cite{Palmgren2007} proves the finitery case, where \( \kappa = \aleph_0 \).  It is natural to expect that their argument generalises to arbitrary regular cardinal \( \kappa \), although we have not yet confirmed it.}
Here we give a different proof, which is syntactic and applicable to arbitrary regular cardinal \( \kappa \).
\begin{theorem}\label{thm:pre:horn-presentable}
    Let \( \Sigma \) be an \( S \)-sorted signature of \( \kappa \)-ary algebra and \( \Th \) be a \( \kappa \)-ary Horn theory.
    The category \( \PAlg(\Sigma, \Th) \) is locally \( \kappa \)-presentable.
\end{theorem}
\begin{proof}(Sketch)
    We reduce a partial Horn theory to an essentially algebraic theory, without changing the category of models (up to equivalence).
The idea is to introduce a new sort \( \mathit{dom}_{\sigma} \) for each operator \( \sigma \in \Sigma \) of sort \( \prod_{i < \alpha} s_i \longrightarrow s \) and to regard \( \sigma \) as a total operator of sort \( \mathit{dom}_{\sigma} \longrightarrow s \).
    The new signature \( \Sigma' \) has total operators \( \pi_{\sigma,i} \colon \mathit{dom}_{\sigma} \longrightarrow s_i \) and the new theory \( \Th' \) has \( \forall x,y \colon \mathit{dom}_\sigma. \bigwedge_{i < \alpha} \pi_{\sigma,i}(x) = \pi_{\sigma,i}(y) \Longrightarrow x = y \), which requires that \( A'_{\mathit{dom}_{\sigma}} \) is (identified with) a subset of \( \prod_{i < \alpha} A'_{s_i} \) for every model \( A' \) of \( (\Sigma',\Th') \).
    If the original theory \( \Th \) requires \( \IsDef{\sigma(\vec{x})} \) under a certain condition, say \( \varphi(\vec{x}) \), we introduce a new partial operator \( \tau_{\varphi,\sigma} \colon \prod_{i < \alpha} s_i \longrightarrow \mathit{dom}_{\sigma} \): its domain is characterised by \( \varphi(\vec{x}) \), and it must satisfy \( \forall \vec{x}. \IsDef{\tau_{\varphi,\sigma}(\vec{x})} \Longrightarrow \pi_i(\tau_{\varphi,\sigma}(\vec{x})) = x_i \), where \( x_i \) is the \( i \)-th variable in \( \vec{x} \).
\end{proof}

\subsection{Relevance to Lafont Models of Linear Logic}
We are interested in locally presentable categories because it gives us a way to construct a \emph{linear exponential comonad} almost for free.
Constructions of models of linear logic using locally presentability and/or accessibility can be date back to Barr's work~\cite{Barr1991}.
The following result can be found in \cite{Porst2008}.
\begin{theorem}\label{thm:pre:presentable-exponential}
    Let \( (\cat, \otimes) \) be a symmetric monoidal closed category and suppose that \( \cat \) is locally presentable.
    Then the forgetful functor \( \Comon(\cat, \otimes) \longrightarrow \cat \) from the category of cocommutative comonoids over \( (\cat, \otimes) \) has a right adjoint.
\end{theorem}
Hence a symmetric-monoidal-closed locally-presentable category is a Lafont model of intuitionistic linear logic (see, e.g., \cite{Mellies2003a} for categorical models of linear logic).

\subsection{Presenting Partial Algebras}
\label{sec:pre:presentation}
This subsection discusses a way to define a partial algebra satisfying a Horn theory \( \Th \) of an \( S \)-sorted signature \( \Sigma \).

A \emph{presentation} comprises of an \( S \)-sorted set \( G = (G_s)_{s \in S} \) of \emph{generators} and an \( S \)-sorted set \( R = (R_s \subseteq \Term^s_{\Sigma}(G) \times \Term^s_\Sigma(G))_{s \in S}\) of \emph{relations}.
A \emph{model} of \( (G,R) \) in \( \PAlg(\Sigma, \Th) \) is a partial algebra \( A \in \PAlg(\Sigma,\Th) \) together with a valuation \( \varrho = (\varrho_s \colon G_s \longrightarrow A_s)_{s \in S} \) such that \( A, \varrho \models t = t' \)
for every \( (t,t') \in \bigcup_{s \in S} R_s \).
Note that, for every model \( (A, (\varrho_s)_{s \in S}) \), a homomorphism \( h\colon A \longrightarrow B \) induces a model \( (B, (h_s \circ \varrho_s)_{s \in S}) \).

A model \( (A, \varrho) \) is \emph{presented by \( (G,R) \) in \( \PAlg(\Sigma,\Th) \)} if, for every model \( (B, \vartheta) \) in \( \PAlg(\Sigma,\Th) \), there is a unique homomorphism \( h \colon A \longrightarrow B \) such that \( \forall s \in S. \vartheta_s = h_s \circ \varrho_s \).
By identifying \( x \in G_s \) with \( \varrho_s(x) \in A_s \), this means that every assignment \( \vartheta \colon G \longrightarrow B \) satisfying \( R \) can be uniquely extended to a homomorphism \( h \colon A \longrightarrow B \).

Every pair \( (G,R) \) of generators and relations has a model that it presents.
Intuitively the model corresponds to the initial object in \( \PAlg(\Sigma \cup G, \Th \cup R) \).
\begin{lemma}\label{lem:partial:presentation}
    For every pair \( (G, R) \) of \( S \)-sorted sets of generators and relations,
there is a model \( (A, \varrho) \) presented by \( (G,R) \).
\end{lemma}
Thus a pair \( (G,R) \) of generators and relations defines a partial algebra with a valuation over \( G \) uniquely up to isomorphism.
We write \( \langle G,R \rangle \) for the partial algebra, which contains \( G \) as elements (but note that different elements \( g, g' \in G_s \) do not necessarily indicate different elements of \( A_s \)). 

\begin{lemma}\label{lem:partial:generated}
    Assume \( A = \langle G,R \rangle \).
    Then \( A \) is generated by \( G \), i.e.~every element \( x \in A_s \) is the interpretation of a term \( t \in \Term^s_\Sigma(G) \).
\end{lemma}

 \section{$\Sigma$-Monoids, Semirings and Modules}
\label{sec:monoid}
This section defines the main object of study of this paper, namely \( \Sigma \)-monoids and corresponding rings and modules.
A \( \Sigma \)-monoid~\cite{Haghverdi2000,Hoshino2012b,Tsukada2018} is an algebra with partially-defined countable sums.
A \( \Sigma \)-semiring and a \( \Sigma \)-module are variants of (semi)ring and of modules,  obtained simply by replacing the category \( \mathbf{Ab} \) of abelian groups with the category \( \SMon \) of \( \Sigma \)-monoids.

After giving definitions, we discuss the relevance of these notions to models of linear logic by examples.

\subsection{$\Sigma$-monoids}

\begin{definition}[$\Sigma$-monoid]\label{def:sigma-monoid}
    A \emph{\( \Sigma \)-monoid} \( \smon \) is a triple \( (|\smon|, 0, \sum) \), where
    \( |\smon| \) is the underlying set,
    \( 0 \in |M| \) is the unit element, and
    \( \sum : |M|^\omega \rightharpoonup |M| \) is a partially-defined countable operation, called the \emph{sum},
subject to the following conditions:
    \begin{enumerate}
        \item for every \( j \in \omega \),
        \begin{equation*}
            \forall (x_i)_i.
            \qquad
            \bigwedge_{i \neq j} x_i = 0
            \quad\Longrightarrow\quad
            \sum_{i \in \omega} x_i = x_j,
        \end{equation*}
        \item for all subsets \( I, J \subseteq \omega \) and every bijection \( \sigma \colon (\omega \setminus I) \longrightarrow (\omega \setminus J) \),
        \begin{align*}
            \forall (x_i)_i, (y_j)_j.
            \bigwedge_{i \in I} x_i = 0
            \wedge
            \bigwedge_{j \in J} y_j = 0
            \wedge
            \bigwedge_{i \in \omega \setminus I} x_i = y_{\sigma(i)}
            \\
            \quad\Longrightarrow\quad
            \sum_{i \in \omega} x_i \Keq \sum_{j \in \omega} y_j,
        \end{align*}
        \item \( \sum_{i \in \omega} x_i \Keq \sum_{j \in \omega} \sum_{k \in \omega} x_{\sigma(j,k)} \) for every bijection \( \omega \times \omega \longrightarrow \omega \),
    \end{enumerate}
    where the sum \( \sum (x_i)_{i \in \omega} \) is written as \( \sum_{i \in \omega} x_i \).
A \( \Sigma \)-monoid is \emph{complete} if
\( \sum_{i \in \omega} x_i \) is defined for every
\( (x_i)_{i \in \omega} \in |\smon|^{\omega} \).
\qed
\end{definition}
Note that the axioms form an \( \aleph_1 \)-ary partial Horn theory.
The first and second axioms say that \( 0 \) is the unit element and that the sum is independent of the order of elements.
The third axiom is the infinite association law.
In particular, if \( \sum_i x_i \) is defined, so is \( \sum_k x_{\sigma(j,k)} \) for every \( k < \omega \).

For a countably infinite set \( I \) and a family \( (x_i)_{i \in I} \), we write \( \sum_{i \in I} x_i \) to mean \( \sum_{j \in \omega} x_{\sigma(j)} \) for some bijection \( \sigma \colon \omega \to I \), which is independent of the choice of \( \sigma \) by the second axiom.
The combination of the second and third axioms shows that, if \( \sum_{i < \omega} x_i \) is defined, so is \( \sum_{i \in I} x_i \) for every subset \( I \subseteq \omega \).

Finite sums are defined via padding: \( x + y + z \) is an abbreviation of the infinite sum \( x + y + z + 0 + 0 + 0 + \cdots \).

Note that even finite sums are not necessarily defined.
A \( \Sigma \)-monoid is \emph{finitely complete} if all finite sums are defined.

\begin{example}\label{ex:monoid:monoid}
    Here are basic examples of \( \Sigma \)-monoids.
    \begin{itemize}
    \item
    \( \CohRing \) has \( \{ 0, 1 \} \) as the underlying set and \( 1 + 1 \) is undefined.
\item
    \( \BoolRing \) has \( \{ 0, 1 \} \) as the underlying set.
    All sums are defined in \( \BoolRing \) and the sum is \( 1 \) if the family contains \( 1 \).
    This is the two-element Boolean algebra \( 0 < 1 \) with the join as the sum.
    This is an example of complete \( \Sigma \)-monoid.
\item
    \( \FinRing \) has \( \{ 0, 1 \} \) as the underlying set.
    It has all finite sums and \( 1 + 1 + \dots + 1 = 1 \).
    The infinite sum \( 1 + 1 + \cdots \) of \( 1 \) is undefined.
    This is a finitely complete \( \Sigma \)-monoid.
\item
    \( \{ 0,1,\infty \} \) with the sum defined as follows, which is complete.
    The sum is \( \infty \) if the family contains \( \infty \) or infinitely many \( 1 \)'s.
    The sum of finitely many \( 1 \)'s is \( 1 \).
\item
    The unit interval \( [0,1] \) with the standard sum of real numbers.
    The sum is undefined if it exceeds \( 1 \).
    This induces the probabilistic coherence space model.
\item
    The sets \( \mathbb{N} \) of natural numbers and \( \mathbb{R}_{\ge 0} \) of non-negative real numbers with the standard sums.
    The sum is undefined if it diverges.
    \item
    \( \mathbb{N} \cup \{ \infty \} \) with the standard sum, where \( \infty \) means the divergence to \( +\infty \).
    Similarly \( \mathbb{R}_{\ge 0} \cup \{ \infty \} \) with the standard sum is a \( \Sigma \)-monoid.
    They are complete.
    \item
    A coherence space~\cite{Girard1987} induces a \( \Sigma \)-monoid.
    A coherence space \( A \) is equipped with a subset \( \Config(A) \subseteq \PowerSet{|A|} \) of cliques, and it can be seen as a \( \Sigma \)-monoid with the disjoint union as the sum.
    \item
    A probabilistic coherence space~\cite{Girard2004,Danos2011}, which is a subset \( A \subseteq (|A| \to \mathbb{R}_{\ge 0}) \) that satisfies certain conditions, is a \( \Sigma \)-monoid with the coordinate-wise sum; it is undefined if the sum does not belong to \( A \).
    \item 
    A poset \( A = (|A|, \le) \) with the join as the sum satisfies the axioms of \( \Sigma \)-monoid if and only if it has all bounded countable joins, i.e.~every countable subset \( X \subseteq |A| \) that has an upper bound in \( |A| \) has the least upper bound.
    \end{itemize}
\end{example}

\begin{example}
    Let \( \mathbb{Q}_{\ge 0} \) be the set of non-negative rational numbers.
    We define \( \sum_{i<\omega} r_i = r \) if the sum converges to \( r \).
    The triple \( (\mathbb{Q}_{\ge 0}, 0, \sum) \) is not a \( \Sigma \)-monoid because it violates the definedness of partial sums: there exists an infinite sequence \( (r_i)_{i < \omega} \) such that \( \sum_i r_i = 2 \) and \( \sum_i r_{2i} = \sqrt{2} \), of which the latter is a partial sum of the former but undefined.
\end{example}

\begin{example}\label{ex:monoid:Kothe}
    Let \( \mathbb{R} \) be the set of reals equipped with the sum defined by the absolute convergence: \( r = \sum_{i < \omega} r_i \) if the right-hand-side absolutely converges to \( r \).
    This is related to the K\"othe sequence space model~\cite{Ehrhard2002a} of linear logic.
    Unfortunately \( \mathbb{R} \) is not a \( \Sigma \)-monoid because the right-to-left direction of the condition (3) in the definition does not hold: \( (1 - 1) = 0 \) and \( (1 - 1) + (1 - 1) + \cdots = 0 + 0 + \cdots \) absolutely converges to \( 0 \) but \( 1 - 1 + 1 - 1 + \cdots \) does not.
\end{example}

Given a \( \Sigma \)-monoid \( \smon \) and \( x,y \in |\smon| \), we write \( x \le y \) to mean \( \exists z. x + z = y \).
This is a preorder, which we call the \emph{associated preorder}.
If \( x \le 0 \), i.e.~\( x + y = 0 \) for some \( y \), then \( x = 0 \) since \( x = x + 0 + 0 + \cdots = x + (y + x) + (y + x) + \cdots = (x + y) + (x + y) + \cdots = 0 + 0 + \cdots =0\).
In this sense, every element of a \( \Sigma \)-monoid except for \( 0 \) is positive.

\begin{definition}[Linear map]
    A \emph{linear map} from a \( \Sigma \)-monoid \( \smon \) to a \( \Sigma \)-monoid \( \smonb \) is a function \( f \colon |\smon| \to |\smonb| \) between the underlying sets such that \( f(0) = 0 \) and \( f(\sum_i x_i) \Kle \sum_i f(x_i) \).
    We write \( \SMon \) for the category of \( \Sigma \)-monoids and linear maps.
\end{definition}
\( \SMon \) is locally presentable, simply because the set of the axioms for \( \Sigma \)-monoids (Definition~\ref{def:sigma-monoid}) is an \( \aleph_1 \)-ary partial Horn theory and linear maps are homomorphisms of models.

The set of linear maps from \( \smon \) to \( \smonb \) is again a \( \Sigma \)-monoid: if \( \sum_i f_i(x) \) converges for every \( x \in |\smon| \), the sum \( \sum_i f_i \) is defined and \( (\sum_i f_i)(x) = \sum_i f_i(x) \).
The zero is the constant function to \( 0 \) of \( \smonb \).
We write \( \smon \multimap \smonb \) for this \( \Sigma \)-monoid.
The category \( \SMon \) is a symmetric monoidal closed category with \( \multimap \) as the linear function space (see \cite{Hoshino2012b}).
The tensor product is the representing object of bilinear maps (see Section~\ref{sec:module:tensor}).
The tensor unit is \( \CohRing \).
\begin{theorem}\label{thm:module:sigma-monoid-Lafont}
    The category \( \SMon \) is a Lafont model of intuitionistic linear logic.
\end{theorem}
\begin{proof}
    Similar to the proof of Theorem~\ref{thm:module:model}.
\end{proof}

\subsection{$\Sigma$-semirings}
A \emph{\( \Sigma \)-semiring} is a ring-like algebra of which the sum is a \( \Sigma \)-monoid.
It is a commutative monoid object in the monoidal category \( (\SMon, \otimes) \); it can also be seen as a single-object version of \( \SMon \)-enriched SMCC studied in \cite{Tsukada2018}.
Here we give a direct definition of \( \Sigma \)-semiring.
\begin{definition}[$\Sigma$-semiring]
    A \emph{\( \Sigma \)-semiring} (or \emph{\( \Sigma \)-rig}) is a tuple \( \srig = (|\srig|, 0, \sum, 1, \cdot) \) where
    \begin{itemize}
        \item \( (|\srig|, 0, \sum) \) is a \( \Sigma \)-monoid,
        \item \( ({-} \cdot {-}) \colon |\srig| \times |\srig| \longrightarrow |\srig| \) is a total function, and
        \item \( 1 \in |\srig| \) is the unit element of \( \cdot \), i.e.~\( 1 \cdot x = x \cdot 1 = x \),
    \end{itemize}
    subject to associativity (\( (x \cdot y) \cdot z = x \cdot (y \cdot z) \)), commutativity (\( x \cdot y = y \cdot x \)), and distributivity, i.e.,
    \begin{equation*}
        \textstyle
        (\sum_i x_i) \cdot (\sum_j y_j) \quad\Kle\quad \sum_i \sum_j (x_i \cdot y_j)
        \quad\mbox{and}\quad
        0 \cdot x = 0.
    \end{equation*}
    Note that the multiplication \( x \cdot y \) is a total operation, whereas \( \sum \) is partial.
    We often omit \( \cdot \) and simply write \( x \cdot y \) as \( xy \).
    A \( \Sigma \)-semiring is \emph{(finitely) complete} if its sum is (finitely) complete.
\end{definition}

\begin{example}\label{ex:module:semiring}
    Here are examples of \( \Sigma \)-semirings.
    \begin{itemize}
    \item
    \( \CohRing \) is a \( \Sigma \)-semiring with the product defined by \( 1 \cdot 1 = 1 \) and \( 0 \cdot 0 = 0 \cdot 1 = 1 \cdot 0 = 0 \).
    (This is actually the canonical monoid structure arising from the canonical isomorphism \( \CohRing \otimes \CohRing \cong \CohRing \) on the tensor unit \( \CohRing \) in \( \SMon \).) 
    \item
    \( \BoolRing \) and \( \FinRing \), which has \( \{0,1\} \) as the underlying set, can be seen as \( \Sigma \)-semirings by the same definition of the product as above.
    \item
    \( \Nat \cup \{ \infty \} \) is a \( \Sigma \)-semiring.
    The product of natural numbers is standard, and \( \infty \cdot 0 = 0 \) and \( \infty \cdot x = \infty \) if \( x \neq 0 \).
    \item
    Continuous semirings, studied in \cite{Laird2013,Laird2016}, are examples of complete \( \Sigma \)-semirings.
    \item
    The unit interval \( [0,1] \) with the standard sum and product of real numbers is a \( \Sigma \)-semiring.
    \item
    The sets \( \mathbb{N} \) of natural numbers and \( \mathbb{R}_{\ge 0} \) of non-negative real numbers with the standard sums and products are \( \Sigma \)-semirings.
\end{itemize}
\end{example}

\iflong
\begin{remark}
    \( \Sigma \)-semirings can be organised into a category, although we shall not use it in this paper.
    A \emph{\( \Sigma \)-semiring homomorphism} \( f \colon \srig \longrightarrow \srigb \) is a \( \Sigma \)-monoid homomorphism between the underlying \( \Sigma \)-monoids that preserves the product and the unit: \( f(1) = 1 \) and \( f(x \cdot y) = f(x) \cdot f(y) \).
    Let \( \SRig \) be the category of \( \Sigma \)-semirings and \( \Sigma \)-semiring homomorphisms.
    It has a symmetric monoidal structure in which the linear function space \( \srig \multimap \srigb \) consists of \( \Sigma \)-semiring homomorphisms with point-wise sum and product.
It is locally \( \aleph_1 \)-presentable since the set of axioms is an \( \aleph_1 \)-ary partial Horn theory.
    Hence \( \SRig \) is a Lafont model of intuitionistic linear logic.    
\end{remark}
\fi

\subsection{Modules over $\Sigma$-semirings}
A (\( \Sigma \)-)module over a \( \Sigma \)-semiring \( \srig \) is a \( \Sigma \)-monoid \( \smod \) with an action \( \srig \otimes \smod \longrightarrow \smod \) satisfying certain conditions.
\begin{definition}[$\srig$-module and $\srig$-linear map]
    Let \( \srig \) be a \( \Sigma \)-semiring.
    A \emph{module over \( \srig \)} (or an \emph{\( \srig \)-module}) is a \( \Sigma \)-monoid \( \smod \) equipped with a bilinear \emph{action} \( ({-} \cdot {-}) \colon |\srig| \times |\smod| \longrightarrow |\smod| \) such that \( (r r') \cdot x = r \cdot (r' \cdot x) \) and \( 1 \cdot x = x \) for every \( r, r' \in |\srig| \) and \( x \in |\smod| \).
    Note that the action is total.
    The bilinearity means \( 0 \cdot x = r \cdot 0 = 0 \) and \( (\sum_i r_i) \cdot (\sum_j x_i) \Kle \sum_i \sum_j (r_i \cdot x_j) \) for all families \( (r_i)_i \) and \( (x_j)_j \) of \( |\srig| \) and \( |\smod| \), respectively.
    An \( \srig \)-module is \emph{(finitely) complete} if so is its sum.
For \( \srig \)-modules \( \smod \) and \( \smodb \), an \emph{\( \srig \)-linear map} \( \smod \longrightarrow \smodb \) is a \( \Sigma \)-monoid homomorphism \( f \) that preserves the action, i.e.~\( f(r \cdot x) = r \cdot f(x) \) for every \( r \in |\srig| \) and \( x \in |\smod| \).
We write \( \SMod[\srig] \) for the category of \( \srig \)-modules and \( R \)-linear maps.
\end{definition}

\begin{example}\label{ex:module:module}
    Here are examples of modules.
    \begin{itemize}
        \item
        For every \( \Sigma \)-semiring \( \srig \), \( \srig \) is an \( \srig \)-module whose action is the product.
        \item 
        Every \( \Sigma \)-monoid is canonically an \( \CohRing \)-module.
        The action is completely determined by the axioms: \( 0 \cdot x = 0 \) and \( 1 \cdot x = x \).
        Since \( 1 + 1 \) is undefined in \( I \), the distributive law \( (\sum_i r_i) \cdot (\sum_j x_j) \Kle \sum_i \sum_j (r_i \cdot x_j) \) requires nothing.
        Hence \( \SMod[\CohRing] \) is isomorphic to \( \SMon \).
        \item
        A countable coherence space can be seen as an \( \CohRing \)-module; see Section~\ref{sec:monoid:example:coherence}.
\item
        A countable finiteness space can be seen as an \( \FinRing \)-module; see Section~\ref{sec:monoid:example:finiteness}.
        \item
        A probabilistic coherence space can be seen as a \( [0,1] \)-module; see Section~\ref{sec:monoid:example:probabilistic}.
\item
        Let \( \srig \) be a complete \( \Sigma \)-semiring.
        For a countable set \( X \), let \( V_X \) be the set \( X \to |\srig| \).
        Then \( V_X \) is an \( \srig \)-module with the point-wise sum and action: \( (\sum_i v_i)(x) \defe \sum_i v_i(x) \) and \( (r \cdot v)(x) \defe r \cdot v(x) \) for \( v_i, v \in V_X \).
For countable sets \( X \) and \( Y \), an \( \srig \)-linear map \( h \colon V_X \longrightarrow V_Y \) is completely determined by the ``matrix'' \( (h_{x,y})_{x \in X, y \in Y} \): \( h(v)(y) = \sum_{x \in X} h_{x,y} v(x) \).
        Conversely every ``matrix'' \( (h_{x,y})_{x \in X, y \in Y} \), \( h_{x,y} \in |\srig| \), defines an \( \srig \)-linear map \( V_X \longrightarrow V_Y \) since the sum \( \sum_{x \in X} h_{x,y} v(x) \) in \( \srig \) is always defined.
        The full subcategory of \( \SMod[\srig] \) consisting of \( V_X \) coincides with the category of weighted relations~\cite{Laird2013,Laird2016}.
        \item
        The \emph{zero module} consisting only of \( 0 \) is an \( \srig \)-module for every \( \srig \).
        It is a zero object of \( \SMod[\srig] \), i.e.~an object that is both initial and terminal.
\end{itemize}    
\end{example}

\begin{example}
    A \( \Sigma \)-monoid \( \smon \) is not necessarily a \( \BoolRing \)-module.
    Since \( 1 + 1 = 1 \) in \( \BoolRing \), every \( \BoolRing \)-module \( \smon \) must satisfy \( x = 1 \cdot x = (1 + 1) \cdot x = 1 \cdot x + 1 \cdot x = x + x \) in \( \smon \) (in particular, \( x + x \) is defined), but this equation does not hold for general \( \Sigma \)-modules.
    Actually \( \BoolRing \)-modules coincide with posets that have all bounded countable joins.
\end{example}

\subsection{Example: Coherence space}\label{sec:monoid:example:coherence}
Let us examine the motivating example, coherence spaces.
\begin{definition}[Coherence space]
    A \emph{coherence space} is a pair \( A = (|A|, \coh_A) \) of a set \( |A| \) of \emph{atoms} and a reflexive relation \( ({\coh}) \subseteq |X| \times |X| \).
    The set \( \Config(A) \) of \emph{cliques} of \( A \) is \( \{ x \subseteq |A| \mid \forall a, a' \in x. a \coh_A a' \} \).
    The relation \( \incoh_A \) is defined by \( a \incoh_A a' \defp a = a' \vee \neg (a \coh_A a') \).
    Given coherence spaces \( A \) and \( B \), the \emph{linear function space} \( A \multimap B \) is defined by \( |A \multimap B| \defe |A| \times |B| \) and \( (a,b) \coh_{A \multimap B} (a', b') \defp \big(a \coh_A a' \Rightarrow b \coh_B b'\big) \wedge \big( b \incoh_B b' \Rightarrow a \incoh_A a' \big) \).
    A \emph{linear map} \( f \colon A \longrightarrow B \) between coherence spaces is a relation \( f \subseteq |A| \times |B| \) such that \( f \in \Config(A \multimap B) \).
    Linear maps are composed as relations.
\end{definition}

Let \( \CCoh \) be the category of \emph{countable} coherence spaces and linear maps.
Here a coherence space \( A \) is countable if \( |A| \) is countable.
We need this restriction because a \( \Sigma \)-monoid has at most countable sums.\footnote{If you would like to consider coherence spaces of larger cardinality, you can choose a sufficiently large regular cardinal \( \kappa \), replace the category of \( \Sigma \)-monoid with the category of partial algebars with \( \kappa \)-ary sums \( \sum_{i < \kappa} x_i \), and adapt the definitions of semirings and modules accordingly.}

We show that a coherence space is an \( \CohRing \)-module.
Mathematically we give a fully faithful functor \( F \colon \CCoh \longrightarrow \SMod[\CohRing] \).
For a coherence space \( A \), \( FA \) is the set \( \Config(A) \) of cliques with the disjoint union as the sum.
The action of \( \CohRing \) is the canonical one: \( 1 \cdot x = x \) and \( 0 \cdot x = \emptyset \).
For a linear map \( f \colon A \longrightarrow B \) in \( \CCoh \), \( Ff \colon FA \longrightarrow FB \) is defined by \( (Ff)(x) \defe \{ b \in |B| \mid \exists a \in x. (a,b) \in f \} \).

\begin{remark}
    The incoherence relation \( \incoh_A \) has a clear interpretation via this embedding: \( a \incoh_A a' \) if and only if \( \IsUndef{(\{a\}\uplus\{a'\})} \) in \( FA \).
\end{remark}
\begin{proposition}
    \( F \) is a fully faithful functor.
\end{proposition}
\begin{proof}
    Let \( f \colon A \longrightarrow B \) be a linear map in \( \CCoh \).
    It is easy to see that \( (Ff)(x) \in \Config(B) \) for every \( x \in \Config(A) \).
    We show that \( Ff \colon FA \longrightarrow FB \) is an \( \CohRing \)-linear map.
    Let \( (x_i)_{i \in I} \) be a countable family in \( \Config(A) \) and suppose that \( \IsDef{(\biguplus_{i \in I} x_i)} \).
    Then, by definition of \( Ff \), \( (Ff)(\biguplus_{i \in I} x_i) = (Ff)(\bigcup_{i \in I} x_i) = \bigcup_{i \in I} (Ff)(x_i) \).
    So it suffices to show that \( (Ff)(x_i) \cap (Ff)(x_{i'}) = \emptyset \) if \( i \neq i' \).
    Assume for contradiction that \( b \in (Ff)(x_i) \cap (Ff)(x_{i'}) \) for some \( i \neq i' \).
    Hence there exist \( a \in x_i \) and \( a' \in x_{i'} \) such that \( (a,b), (a',b) \in f \).
    Since \( (a,b) \coh_{A \multimap B} (a',b) \) and \( b \incoh_B b \), we have \( a \incoh_A a' \).
    Since \( i \neq i' \), we have \( x_i \cap x_{i'} = \emptyset \) and thus \( a \neq a' \).
    Hence \( \neg (a \coh_A a') \) but this contradicts \( x_i \uplus x_{i'} \in \Config(A) \).

    \( F \) is obviously faithful.
    We prove that \( F \) is full.
    Let \( g \colon FA \longrightarrow FB \) be an \( \CohRing \)-linear map.
    We define \( f \subseteq |A| \times |B| \) by \( (a,b) \in f \defp b \in g(\{a\}) \).
    We show that \( f \in \Config(A \multimap B) \).
    Suppose that \( (a,b), (a',b') \in f \).
    \begin{itemize}
        \item If \( a = a' \), then \( b,b' \in g(\{a\}) \in \Config(B) \) and thus \( b \coh_B b' \).
        \item If \( a \neq a' \) and \( a \coh_A a' \), then \( \{ a \} \uplus \{ a' \} = \{ a,a' \} \in \Config(A) \).  So \( \{ b,b' \} \subseteq g(\{a\}) \uplus g(\{a'\}) = g(\{a,a'\}) \in \Config(B) \), i.e.~\( b \coh_B b' \).
        \item We prove the contraposition.  Assume \( \neg(a \incoh_A a') \), i.e.~\( \IsDef{(\{a\}\uplus\{a'\})} \).
        Then
        \( \Config(B) \ni g(\{a\}\uplus\{a'\}) = g(\{a\})\uplus g(\{a'\}) \supseteq \{ b \} \uplus \{ b' \} \).
        So \( \IsDef{(\{b\}\uplus\{b'\})} \), i.e.~\( \neg (b \incoh_B b') \).
    \end{itemize}

    Te prove \( Ff = g \), let \( x \in \Config(A) \).
    Then \( x = \biguplus_{a \in x} \{ a \} \) and thus \( g(x) = \biguplus_{a \in x} g(\{a\}) = \bigcup_{a \in x} (Ff)(\{a\}) = (Ff)(x) \).
\end{proof}

\subsection{Example: Finiteness spaces}\label{sec:monoid:example:finiteness}
\begin{definition}[Finiteness space~\cite{Ehrhard2005a}]
    For a set \( I \) and a subset \( X \subseteq \PowerSet{I} \), its \emph{dual} \( X^{\bot} \) is \( \{ u \in \PowerSet{I} \mid \forall x \in X. \#(x \cap u) < \infty \} \), where \( \# A < \infty \) means \( A \) is a finite set.
    A \emph{finiteness space} \( A \) is a pair \( A = (|A|, \FConfig(A)) \) of a set \( |A| \) and \( \FConfig(A) \subseteq \PowerSet{|A|} \) such that \( \FConfig(A)^{\bot\bot} = \FConfig(A) \).
    A morphism \( f \colon A \longrightarrow B \) between finiteness spaces, called a \emph{finitary relation}, is a relation \( f \subseteq |A| \times |B| \) such that \( \forall x \in \FConfig(A). \{ b \in |B| \mid \exists a \in x. (a,b) \in f \} \in \FConfig(B) \) and \( \forall v \in \FConfig(B)^{\bot}. \{ a \in |A| \mid \exists b \in v. (a,b) \in f \} \in \FConfig(A)^{\bot} \).
    Morphisms are composed as relations.
\end{definition}

Let \( \CFin \) be the category of \emph{countable} finiteness spaces and finitary relations.
Here a finiteness space \( A \) is countable if \( |A| \) is a countable set.

We give a fully faithful functor \( G \colon \CFin \longrightarrow \SMod[\FinRing] \).
Given a finiteness space \( A \), \( GA \) is an \( \FinRing \)-module with the underlying set \( \FConfig(A) \).
For a countable family \( (x_i)_{i \in I} \) on \( \FConfig(A) \), the sum \( \sum_{i \in I} x_i \) is defined if and only if (1) \( \# \{ i \in I \mid a \in x_i \} < \infty \) for every \( a \in |A| \) and (2) \( \big(\bigcup_{i \in I} x_i\big) \in \FConfig(A) \); then \( \sum_{i \in I} x_i = \bigcup_{i \in I} x_i \).
Note that a finite sum \( x_1 + \dots + x_n \) is always defined.
Given a finitary relation \( f \colon A \longrightarrow B \) in \( \CFin \), \( (Gf)(x) \defe \{ b \in |B| \mid \exists a \in x. (a,b) \in f \} \).

\begin{lemma}
    Let \( A \) be a finiteness space.
    Then \( v \in \FConfig(A)^{\bot} \) bijectively corresponds to an \( \FinRing \)-linear map \( v^{\dagger} \colon GA \longrightarrow \FinRing \), defined by \( v^{\dagger}(x) = 1 \) iff \( x \cap v \neq \emptyset \).
\end{lemma}
\begin{proof}
    We see the \( \FinRing \)-linearity of \( v^{\dagger} \).
    Suppose for contradiction that \( v \in \FConfig(A)^{\bot} \) and \( \IsDef{(\sum_{i \in I} x_i)} \) in \( GA \) but \( \IsUndef{(\sum_{i \in I} v^{\dagger}(x_i))} \).
    Then \( v^{\dagger}(x_i) = 1 \), i.e.~\( x_i \cap v \neq \emptyset \), for infinitely many \( i \in I \).
    Since \( \bigcup_{i \in I} x_i \in \FConfig(A) \) from \( \IsDef{(\sum_{i \in I} x_i)} \), the intersection \( (\bigcup_{i \in I} x_i) \cap v \) is a finite set.
    Hence some \( a \in v \) appears infinitely many times in \( (x_i)_{i \in I} \), which contradicts \( \IsDef{(\sum_{i \in I} x_i)} \).
    It is easy to see that \( v^\dagger(\sum_{i \in I} x_i) = \sum_{i \in I} v^\dagger(x_i) \) once the definedness of the right-hand-side is established.

    The map \( ({-})^\dagger \) is obviously injective.
    We prove that it is surjective.
    Given \( h \colon GA \longrightarrow \FinRing \), let \( v_h \defe \{ a \in |A| \mid h(\{a\}) = 1 \} \).
    For every \( x \in \FConfig(A) \), we have \( h(x) = h(\sum_{a \in x} \{a\}) \Kle \sum_{a \in x} h(\{a\}) \).
    The definedness of the right most equation implies \( \forall x \in \FConfig(A). \#(x \cap v_h) < \infty \), and hence \( v_h \in \FConfig(A)^\bot \).
    It is easy to see that \( v_h^\dagger(x) = h(x) \). 
\end{proof}

\begin{proposition}
    \( G \) is a fully faithful functor.
\end{proposition}
\begin{proof}
    We first prove that \( Gf \) is \( \FinRing \)-linear for a finitary relation \( f \colon A \longrightarrow B \).
    Suppose that \( x = \sum_{i\in I} x_i \) in \( GA \) and assume for contradiction that \( \IsUndef{(\sum_{i \in I} (Gf)(x_i))} \).
    Note that \( \bigcup_i (Gf)(x_i) = (Gf)(\bigcup_i x_i) = (Gf)(x) \in \FConfig(B) \) by the definition of \( Gf \).
    Hence there exists \( b \in |B| \) such that \( \# \{ i \in I \mid b \in (Gf)(x_i) \} = \infty \).
    By removing \( x_j \) from the sum \( \sum_i x_i \) if \( b \notin (Gf)(x_j) \), we can assume without loss of generality that \( b \in (Gf)(x_i) \) for every \( i \in I \).
    Let us choose \( a_i \in x_i \) such that \( (a_i,b) \in f \) for each \( i \in I\) and \( \equiv \) be an equivalence relation on \( I \) defined by \( i \equiv j \defp a_i = a_j \).
    Since \( \sum_{i \in I} x_i \) is defined, so is \( \sum_{i \in I} \{ a_i \} \).
    Hence each equivalence class of \( \equiv \) is finite.
    Because \( I \) is infinite (otherwise \( \sum_{i \in I} (Gf)(x_i) \) is a finite sum and hence defined), the number of equivalence classes is infinite.
    Hence \( \sum_{i \in I} \{ a_i \} = \{ a_i \mid i \in I \} \) is infinite.
    Since \( \{ b \} \in \FConfig(B)^{\bot} \) (as every finite set belongs to \( \FConfig(B)^{\bot} \)), \( \{ a \in |A| \mid (a, b) \in f \} \in \FConfig(A)^{\bot} \) by the definition of finitely relation.
    However \( \{ a_i \mid i \in I \} \in \FConfig(A) \) and \( \{ a_i \mid i \in I \} \cap \{ a \in |A| \mid (a,b) \in f \} = \{ a_i \mid i \in I \} \) is an infinite set, a contradiction.

    Faithfulness is clear.
    We prove the fullness.
    Assume a \( \FinRing \)-linear map \( g \colon FA \longrightarrow FB \).
    Let \( f \subseteq |A| \times |B| \) be the relation defined by \( (a,b) \in f \defp b \in g(\{b\}) \) (note that \( \{ b \} \in \FConfig(B) \) as \( \FConfig(B) \) contains all finite subsets of \( |B| \)).
    For every \( x \in \FConfig(A) \), we have \( g(x) = g(\sum_{a \in x} \{a\}) \Kle \sum_{a \in x} g(\{a\}) \Kle \{ b \in |B| \mid \exists a \in x. (a,b) \in f \} \).
    So \( \{ b \in |B| \mid \exists a \in x. (a,b) \in f \} \in \FConfig(B) \) since \( g(x) \in \FConfig(A) \).
    Let \( v \in \FConfig(B)^{\bot} \).
    Then \( v^{\dagger} \) is a \( \FinRing \)-linear map \( GB \longrightarrow \FinRing \).
    Hence \( (v^{\dagger} \circ g) \colon GA \longrightarrow \FinRing \) is \( \FinRing \)-linear since \( \FinRing \)-linear maps compose.
    Hence \( \{ a \in |A| \mid \exists b \in v. (a,b) \in f \} = \{ a \in |A| \mid v^{\dagger}(g(\{a\})) = 1 \} \in \FConfig(A)^{\bot} \).
\end{proof}

\subsection{Example: Probabilistic coherence space}\label{sec:monoid:example:probabilistic}
The probabilistic coherence space model~\cite{Girard2004,Danos2011} is a probabilistic variant of the coherence space model.
It provides us with a fully-abstract model of a higher-order probabilistic programming language~\cite{Danos2011,Ehrhard2014}.
We show that probabilistic coherence spaces are \( [0,1] \)-modules.
\begin{definition}[Probabilistic coherence space]
    Let \( \mathbb{R}_{\ge 0} \) be the set of non-negative real numbers.
    For a countable set \( I \) and a subset \( P \subseteq (\mathbb{R}_{\ge 0})^I \), its \emph{dual} is another subset \( P^{\bot} \subseteq (\mathbb{R}_{\ge 0})^I \) defined by \( P^{\bot} \defe \{ u \in (\mathbb{R}_{\ge 0})^I \mid \forall x \in P. \sum_{i \in I} x(i) u(i) \le 1 \} \), where the sum is the standard sum of reals.
    A \emph{probabilistic coherence space} \( A \) is a pair \( A = (|A|, \PConfig(A)) \) of a countable set \( |A| \) and \( \PConfig(A) \subseteq (\mathbb{R}_{\ge 0})^{|A|} \) that satisfies (1) \( \PConfig(A)^{\bot\bot} = \PConfig(A) \), (2) \( \forall a \in |A|.\, \exists r \in \mathbb{R}.\, \forall x \in \PConfig(A).\: x(a) \le r \), and (3) \( \forall a \in |A|.\, \exists r \in \mathbb{R}.\: r e_a \in \PConfig(A) \), where \( e_a(a) = 1 \) and \( e_a(a') = 0 \) if \( a \neq a' \).
    An \( (|A|\times|B|) \)-matrix \( f = (f_{a,b})_{a \in |A|, b \in |B|} \), \( f_{a,b} \in \mathbb{R}_{\ge 0} \), induces a partial function \( \hat{f} \) from \( (\mathbb{R}_{\ge 0})^{|A|} \) to \( (\mathbb{R}_{\ge 0})^{|B|} \): given \( x \in (\mathbb{R}_{\ge 0})^{|A|} \), if \( \sum_{a \in |A|} f_{a,b} x(a) \) converges in \( \mathbb{R}_{\ge 0} \) for every \( b \in |B| \), then \( \hat{f}(x) \in (\mathbb{R}_{\ge 0})^{|B|} \) is defined by \( \hat{f}(x)(b) \defe \sum_{a \in |A|} f_{a,b} \, x(a) \); otherwise \( f(x) \) is undefined.
    A morphism \( f \colon A \longrightarrow B \) between probabilistic coherence spaces is an \( (|A| \times |B|) \)-matrix \( f = (f_{a,b})_{a \in |A|,b\in |B|} \) over \( \mathbb{R}_{\ge 0} \) such that \( \hat{f}(x) \) is defined and \( \hat{f}(x) \in \PConfig(B) \) for every \( x \in \PConfig(A) \).
    The composite of \( f \colon A \longrightarrow B \) and \( g \colon B \longrightarrow C \) is given by the matrix composition, i.e.~\( (g \circ f)_{a,c} \defe \sum_{b \in |B|} g_{b,c} \, f_{a,b} \).
    We write \( \CPCoh \) for the category of probabilistic coherence spaces.
\end{definition}
\begin{remark}
    A remark at the beginning of \cite[Section~2.1.2]{Danos2011} discusses why not \( \PConfig(A) \subseteq [0,1]^{|A|} \), despite that it ``might seem a desirable (or at least intuitively appealing) condition''.
    This paper would provide a new perspective to this interesting problem: actually we can define \( \CPCoh \) using only \( [0,1] \) in the definition.
    See Example~\ref{eg:classical:proj}.
\end{remark}

We give a fully faithful functor \( H \colon \CPCoh \longrightarrow \SMod[{[0,1]}] \).
Given a probabilistic coherence space \( A \), \( HA \) is \( \PConfig(A) \) with the point-wise sum \( (\sum_{i \in I} x_i)(a) = \sum_{i \in I} x_i(a) \); the sum is undefined if \( \sum_{i \in I} x_i(a) \) diverges for some \( a \) or \( (\sum_i x_i) \notin \PConfig(A) \).
The action of \( r \in [0,1] \) is \( (rx)(a) = r(x(a)) \), which is a total operation on \( \PConfig(A) \) since \( \forall a \in |A|. (rx)(a) \le x(a) \).
For a morphism \( f \colon A \longrightarrow B \) in \( \CPCoh \), \( Hf = \hat{f}{\upharpoonright_{\PConfig(A)}} \), the restriction of \( \hat{f} \) to \( \PConfig(A) \).
\begin{proposition}
    \( H \) is a fully faithful functor.
\end{proposition}
\begin{proof}
    The functoriality is clear.

    We prove the faithfulness.
    Let \( f,g \colon A \longrightarrow B \) in \( \CPCoh \) be different morphisms.
    Then \( f_{a,b} \neq g_{a,b} \) for some \( a \in |A| \) and \( b \in |B| \).
    Since \( A \) is a probabilistic coherence space \( re_a \in \PConfig(A) \) for some \( 0 < r \in \mathbb{R} \), where \( re_a(a) = r \) and \( re_a(a') = 0 \) if \( a \neq a' \).
    Then \( (Hf)(re_a)(b) = f_{a,b} r \neq g_{a,b} r = (Hg)(r e_a)(b) \) since \( r \neq 0 \) and \( f_{a,b} \neq g_{a,b} \).

    We prove the fullness.
    Let \( h \colon HA \longrightarrow HB \) be a \( [0,1] \)-linear map.
    For each \( a \in |A| \), let \( \gamma_a \in (0,\infty) \) be the value defined by \( \gamma_a = \sup \{ x(a) \mid x \in \PConfig(A) \} \).
    It follows from \( \PConfig(A)^{\bot\bot} = \PConfig(A) \) that \( \gamma_a e_a \in \PConfig(A) \).
    Note that \( x = \sum_{a \in |A|} r_a \cdot [\gamma_a e_a] \) holds\footnote{The square bracket in the expression \( [\gamma_a e_a] \) emphasises the fact that \( \gamma_a e_a \) isn \emph{not} the action of \( \gamma_a \) to \( e_a \).
        Actually \( e_a \) may not belong to \( \PConfig(A) \) and \( \gamma_a \) may not belong to \( [0,1] \).
        Hence one \emph{cannot} decompose \( [\gamma_a e_a] \) into \( \gamma_a \cdot e_a \) by using the action in \( HA \) and the equation \( h(\gamma_a e_a) = \gamma_a \cdot h(e_a) \) fails (or even ill-defined).}
    in \( HA \) for every \( x \in \PConfig(A) \), where \( r_a = x(a)/\gamma_a \in [0,1] \).
    Let \( f_{a,b} = (1/\gamma_a) h([\gamma_a e_a])(b) \).
    For every \( x \in \PConfig(A) \) and \( b \in |B| \), since \( h(x) = h(\sum_{a \in |A|} r_a \cdot [\gamma_a e_a]) \Kle \sum_{a \in |A|} r_a \cdot h([\gamma_a e_a]) \),
    \begin{align*}
        & \textstyle
        h(x)(b)
        = \sum_{a \in |A|} r_a h([\gamma_a e_a])(b)
        \\
        & \textstyle \quad
        = \sum_{a \in |A|} x(a) (1/\gamma_a) h([\gamma_a e_a])(b)
        = \sum_{a \in |A|} x(a) f_{a,b}.
    \end{align*}
    Since \( b \) is arbitrary, this means \( \hat{f}(x) \in (\mathbb{R}_{\ge 0})^{|B|} \) is defined.
    Furthermore \( \hat{f}(x) = h(x) \in \PConfig(B) \).
    Since \( x \in \PConfig(A) \) is arbitrary, \( f \colon A \longrightarrow B \) in \( \CPCoh \).
    Trivially \( Hf = h \) as desired.
\end{proof}

 \section{Structures of Categories of Modules}\label{sec:module}
This section studies the structures and constructions about \( \srig \)-modules.
The main theorem of this section shows that \( \SMod[\srig] \) is a model of intuitionistic linear logic with the cofree exponential.

We also give some technical results, used in later sections.

\subsection{Local presentability and Lafont structure}
\label{sec:module:tensor}
Let \( \srig \) be a \( \Sigma \)-semiring.
The category \( \SMod[\srig] \) is locally \( \aleph_1 \)-presentable since it is the category of models of a partial Horn theory (over the signature having \( 0 \), \( \sum \) and \( r \cdot ({-}) \) for each \( r \in |\srig| \) as operations). 
\begin{theorem}
    For every \( \Sigma \)-semiring \( \srig \), the category \( \SMod[\srig] \) of \( \srig \)-modules and \( \srig \)-linear maps is locally \( \aleph_1 \)-presentable.
\end{theorem}

Now it suffices to give a symmetric monoidal closed structure on \( \SMod[\srig] \), which automatically becomes a Lafont model.

The tensor product that we study in this paper is a representing object of \( \srig \)-bilinear maps.
Given \( \srig \)-modules \( \smod \), \( \smodb \) and \( \smodc \), an \emph{\( \srig \)-bilinear map} from \( \smod \) and \( \smodb \) to \( \smodc \) is a function \( f \colon |\smod| \times |\smodb| \longrightarrow |\smodc| \) such that \( f(\sum_i r_i \cdot x_i, \sum_j s_j \cdot y_j) \Kle \sum_{i,j} r_i s_j \cdot f(x_i, y_j) \).
Let \( \Bilin_{\srig}(\smod, \smodb; \smodc) \) be the set of \( \srig \)-bilinear maps.
It defines a functor \( \Bilin_{\srig}(\smod, \smodb; {-}) \colon \SMod[\srig] \longrightarrow \Set \), where \( \Bilin_{\srig}(\smod, \smodb; g)(f) \defe g \circ f \).
The representability of \( \Bilin_{\srig}(\smod, \smodb; {-}) \) is proved in \cite{Hoshino2012b}; here we give a more explicit construction
using presentation (see Section~\ref{sec:pre:presentation}).
\begin{lemma}\label{lem:module:tensor-exists}
    \( \Bilin_{\srig}(\smod, \smodb; {-}) \) has a representing object.
\end{lemma}
\begin{proof}
    We give a presentation of a representing object.
    The set \( G \) of generators is \( \{ x \otimes y \mid (x,y) \in |\smod| \times |\smodb| \} \).
    The relation \( R \) has \( x \otimes y = \sum_i r_i \cdot (x_i \otimes y) \) if \( x = \sum_i r_i \cdot x_i \) in \( \smod \) and \( x \otimes y = \sum_j r_j \cdot (x \otimes y_j) \) if \( y = \sum_j r_j \cdot y_j \) in \( \smodb \).
\end{proof}

Let \( \smod \otimes \smodb \) be the representing object of \( \Bilin_{\srig}(\smod, \smodb; {-}) \).
By construction, each element \( z \in |\smod \otimes \smodb| \) is generated by \( \{ x \otimes y \mid (x,y) \in |\smod| \times |\smodb| \} \), i.e.~\( z = \sum_{i} x_i \otimes y_i \) for some countable family \( (x_i \otimes y_i)_{i} \).

Note that \( \sum_i (x_i \otimes y_i) \) is not necessarily defined for every countable family \( (x_i \otimes y_i)_{i} \).
This fact gives rise to an obstacle when defining an element of \( \smod \otimes \smodb \) or a function to \( \smod \otimes \smodb \): one has to ensure that the sum in the definition actually converges.
The following lemma is useful for ensuring the definedness.
\begin{lemma}\label{lem:module:tensor-peak}
    Let \( \smod \) and \( \smodb \) be \( \srig \)-modules.
    For every \( z \in |\smod \otimes \smodb| \), there is \( (x,y) \in |\smod| \times |\smodb| \) such that \( z \le x \otimes y \) w.r.t.~the associated preorder.
\end{lemma}

The right adjoint of \( ({-}) \otimes \smod \) is given by the module of linear maps.
Given a \( \srig \)-modules \( \smod \) and \( \smodb \), \( \smod \multimap \smodb \) is a \( \srig \)-module consisting of \( \srig \)-linear maps \( \smod \longrightarrow \smodb \).
The sum of a family \( (f_i)_{i \in \omega} \) is defined when \( \sum_i f_i(x) \) is defined for every \( x \in |\smon| \); if the sum is defined, \( (\sum_i f_i)(x) \defe \sum_i f_i(x) \).
The action is defined by \( (r \cdot f)(x) \defe r \cdot f(x) \).
Given an \( \srig \)-linear map \( g \colon \smodb \longrightarrow \smodc \), \( \smod \multimap g \colon (\smod \multimap \smodb) \longrightarrow (\smod \multimap \smodc) \) is the post-composition of \( g \).

\begin{lemma}
    \( ({-}) \otimes \smod \dashv \smod \multimap ({-}) \).
\end{lemma}
\begin{proof}
    Because \( \Bilin(\smod, \smodb; \smodc) \cong \SMod[\srig](\smod, \smodb \multimap \smodc) \).
\end{proof}

\begin{theorem}\label{thm:module:model}
    For every \( \Sigma \)-semiring \( \srig \), the category \( \SMod[\srig] \) is a locally \( \aleph_1 \)-presentable symmetric monoidal closed category.
    Hence it is a Lafont model of intuitionistic linear logic.
\end{theorem}

\begin{remark}
    Although we have directly constructed the SMCC structures on \( \SMod[\srig] \) for every \( \srig \), they actually come from the SMCC structure on \( \SMon \)
and a general result as follows.
    Recall that a \( \Sigma \)-semiring \( \srig \) is just a commutative monoid object in \( \SMon \).
    A commutative monoid structure \( \srig \) makes the functor \( \srig \otimes ({-}) \colon \SMon \longrightarrow \SMon \) a commutative monad, and the category \( \SMod[\srig] \) is then the Eilenberg-Moore category of the monad.
    Then \( \SMod[\srig] \) carries an SMCC structure~\cite{Jacobs1994}.
    Since the monad \( \srig \otimes ({-}) \) is \( \aleph_1 \)-accessible (as it has the right adjoint \( \srig \multimap ({-}) \)), \( \SMod[\srig] \) is locally \( \aleph_1 \)-presentable~\cite[Theorem and Remark~2.78]{Adamek1994}.
\end{remark}

As a locally presentable category,
\( \SMod[\srig] \) is complete and cocomplete.
We describe some important (co)limits.

The \emph{zero module} consisting only of \( 0 \) is the zero object.

Given a countable family \( (\smod_i)_{i \in I} \) of \( \srig \)-modules, their product \( \prod_{i \in I} \smod_i \) has \( \prod_{i \in I} |\smod_i| \) as the underlying set.
The action and sum are coordinate-wise, i.e.~\( r \cdot (x_1,x_2,\dots) = (r \cdot x_1, r \cdot x_2, \dots) \) and \( \sum_i (x_{1,i},x_{2,i},\dots) = (\sum_i x_{1,i}, \sum_i x_{2,i}, \dots) \). 
The sum is undefined if it is undefined on at least one coordinate. 

The coproduct \( \coprod_{i \in I} \smod_i \) is the disjoint union of \( |\smod_i| \) followed by identification of \( 0 \)'s in different components.
The sum of non-zero elements in different component is undefined.
Alternatively it can be seen as the submodule (see Definition~\ref{def:module:submodule}) of \( \prod_{i \in I} \smod_i \) consisting of elements \( (x_1,x_2,\dots) \) such that \( x_i = 0 \) for all but at most one \( i \in I \).

\subsection{Monomorphism}
\begin{lemma}\label{lem:module:injection}
    An \( \srig \)-linear map \( f \colon \smod \longrightarrow \smodb \) is a monomorphism if and only if the underlying function is injective.
\end{lemma}
\begin{definition}\label{def:module:submodule}
    Let \( \smod \) be an \( \srig \)-module.
    An \( \srig \)-module \( \smodb \) is a \emph{submodule} of \( \smod \) if \( |\smodb| \subseteq |\smod| \) and the inclusion \( |\smodb| \hookrightarrow |\smod| \) is \( \srig \)-linear.
    There is a bijective correspondence between submodules of \( \smod \) and subobjects of \( \smod \) in \( \SMod[\srig] \).
    A submodule \( \iota \colon \smodb \hookrightarrow \smod \) is \emph{sum-reflecting} if for every \( x, y_i \) in \( |\smodb| \), if \( \iota(x) = \sum_i \iota(y_i) \) in \( \smod \), then \( \sum_i y_i \) is defined in \( \smodb \) (and then \( x = \sum_i y_i \) in \( \smodb \) as well).
    A submodule \( \smodb \) of \( \smod \) is \emph{downward closed} if \( |\smodb| \subseteq |\smod| \) is a downward-closed subset w.r.t.~the associate preorder of \( \smod \).
\end{definition}

\begin{example}
    We have \( |\CohRing| = |\FinRing| = |\BoolRing| = \{ 0,1 \} \) and the identity function defines monomorphisms \( \CohRing \hookrightarrow \FinRing \hookrightarrow \BoolRing \) in \( \SMod[\CohRing] \) (or equivalently in \( \SMon \)), which are epimorphisms but not isomorphisms.
    Neither \( \CohRing \hookrightarrow \FinRing \) nor \( \FinRing \hookrightarrow \BoolRing \) is sum-reflecting.
\end{example}
\begin{example}
    \( \CohRing \hookrightarrow \mathbb{N} \hookrightarrow \mathbb{N}_\infty \) in \( \SMod[\CohRing] \) and
    \( [0,1] \hookrightarrow \mathbb{R}_{\ge 0} \) in \( \SMod[{[0,1]}] \) are sum-reflecting.
\end{example}

The equaliser is easily described as a sum-reflecting submodule.
(The converse does not hold in general.)
\begin{lemma}\label{lem:module:equaliser}
    The equaliser of \(f,g \colon \smod \longrightarrow \smodb \) is a sum-reflecting submodule of \( \smod \) consisting of \( \{ x \in |\smod| \mid f(x) = g(x) \} \).
\end{lemma}

\subsection{Completion}
The partiality of sums is the source of many difficulties in constructing and analysing \( \srig \)-modules.
Here we give a remedy to this.

\begin{lemma}\label{lem:module:completion}
    For every \( \srig \)-module \( \smod \), there is a complete \( \srig \)-module \( \CComp{\smod} \) with an \( \srig \)-linear map \( f \colon \smod \longrightarrow \CComp{\smod} \) that satisfies the following property: for every complete \( \srig \)-module \( \smodb \) and an \( \srig \)-linear map \( g \colon \smod \longrightarrow \smodb \), there is a unique \( \srig \)-linear map \( h \colon \CComp{\smod} \longrightarrow \smodb \) such that \( g = h \circ f \).
\end{lemma}

\begin{lemma}\label{lem:module:completion-monic-epic}\label{lem:module:completion-dcfull}
    \( \smod \longrightarrow \CComp{\smod} \) is monic and epic.
    Furthermore 
    \( \smod \hookrightarrow \CComp{\smod} \) is a sum-reflecting and downward-closed submodule.
\end{lemma}

By this lemma, one can first argue under the assumption that sums in \( e \) is defined, and then discharge the assumption by showing that \( e \le x \) in \( \CComp{\smod} \) for some \( x \in |\smod| \).

 \section{Models of Classical Linear Logic}\label{sec:classic}
The category \( \SMod[\srig] \) of modules is a model of intuitionistic linear logic but not of classical linear logic, i.e.~the canonical map \( \smod \longrightarrow (\smod \multimap \srig) \multimap \srig \) is not necessarily isomorphic.
This section discusses constructions of models of classical linear logic.
The idea is to consider a subcategory consisting of ``countable dimensional'' modules.

There are several choices, corresponding to several well-behaved modules in the standard module theory over rings.
This section discusses two variants that correspond to \emph{finitely-generated free modules} and \emph{finitely-generated projective modules} in the standard module theory over rings.\footnote{Unfortunately the \( \Sigma \)-semiring counterparts of these notions are no longer free nor projective in the categorical sense.}

\subsection{Modules with orthogonal basis}
\emph{Finitely generated free modules} form a class of particularly well-behaved modules in the standard module theory over rings.
A finitely generated free module is the free module over a finite set; it coincides with the finite biproduct \( \bigoplus_{i=1}^n R \) of the coefficient ring \( R \).

So a natural candidate of well-behaved modules over an \( \Sigma \)-semiring \( \srig \) is a finitely or countably generated free modules.
Unfortunately this class does not have good closure properties; they are finite or countable coproducts \( \coprod_{i \in I} \srig \) of the coefficient semiring \( \srig \), which is not closed under many operations of linear logic, e.g.~negation \( ({-}) \multimap \srig \), tensor \( \otimes \) and linear implication \( \multimap \).
Using products instead of coproducts does not change the situation.

Finitely generated free modules have many different characterisations in the standard module theory, which are not necessarily equivalent in our \( \Sigma \)-semiring setting.
The next definition is an adaptation of one of such characterisations.
\begin{definition}[Orthogonal dual basis]
    Let \( \srig \) be a \( \Sigma \)-semiring and \( \smod \) be an \( \srig \)-module.
    An \emph{(countable) orthogonal dual basis} of \( \smod \) is a countable family \( (e_i, \varphi_i)_{i \in I} \) of pairs of \( e_i \in |\smod| \) and \( \varphi_i \colon \smod \longrightarrow \srig \) such that
    \begin{itemize}
        \item \( x = \sum_{i \in I} \varphi_i(x) \cdot e_i \) for every \( x \in |\smod| \) and
        \item \( \varphi_i(e_j) = \delta_{i,j} \), where \( \delta_{i,j} = 1 \) if \( i = j \) and \( \delta_{i,j} = 0 \) if \( i \neq j \).
    \end{itemize}
    We write \( \SVec[\srig] \) for the full subcategory of \( \SMod[\srig] \) consisting of modules that have countable orthogonal dual bases.
\end{definition}

\begin{lemma}
    An \( \srig \)-module \( \smod \) has a countable orthogonal dual basis if and only if there is a countable family \( (e_i)_{i \in I} \) of \( |\smod| \) that satisfies the following conditions:
    \begin{itemize}
        \item every \( x \in |\smod| \) can be uniquely written as \( x = \sum_{i \in I} r_i \cdot e_i \) for some family \( (r_i)_{i \in I} \) of \( |\srig| \), and
        \item the \( e_i \)-coordinate map \( \varphi_i \colon \smod \longrightarrow \srig \) defined by \( (\sum_i r_i \cdot e_i) \mapsto r_i \) is \( \srig \)-linear for every \( i \in I \).
    \end{itemize}
\end{lemma}
\begin{remark}
    The second condition is not redundant.
    For example,
    \( B \) is an \( I \)-module and every element \( x \in B \) can be uniquely written as \( x = r \cdot 1 \), \( r \in I \), but the coordinate map \( \varphi \) is not \( I \)-linear: \( \varphi(1+1) = \varphi(1) = 1 \) but \( \IsUndef{(\varphi(1) + \varphi(1))} \).
\end{remark}

For an \( \srig \)-module \( \smod \) with a countable orthogonal dual basis \( (e_i, \varphi_i)_{i \in I} \), it is not necessarily the case that \( \sum_i r_i \cdot e_i \) is defined for every family \( (r_i)_i \) of \( |\srig| \).
If it is the case, then \( \smod \) is isomorphic to the product of \( \#I \) copies of \( \srig \)'s.
An example of a module that has a countable orthogonal dual basis but is not the product of copies of \( \srig \) is a coproduct of copies of \( \srig \): every element \( x \) of \( \coprod_{i \in I} \srig \) can be uniquely written as \( x = \sum_{i \in I} r_i \cdot e_i \) but the sum \( \sum_{i \in I} r_i \cdot e_i \) is defined only if \( r_i = 0 \) for all but one \( i \in I \).

\begin{example}\label{eg:module:vectors}
    Recall functors \( F \colon \CCoh \hookrightarrow \SMod[\CohRing] \), \( G \colon \CFin \hookrightarrow \SMod[\FinRing] \) and \( H \colon \CPCoh \hookrightarrow \SMod[{[0,1]}] \) defined in Sections~\ref{sec:monoid:example:coherence}--\ref{sec:monoid:example:probabilistic}.
    \begin{enumerate}
        \item \( F A \) has a countable orthogonal dual basis for every countable coherence space \( A \).
        A basis is \( (\{a\}, \varphi_a)_{a \in |A|} \), where \( \varphi_a(x) = 1 \) if \( a \in x \) and \( \varphi_a(x) = 0 \) if \( a \notin x \).
        \item \( G A \) has a countable orthogonal dual basis for every countable finiteness space \( A \).
        A basis is \( (\{a\}, \varphi_a)_{a \in |A|} \), where \( \varphi_a(x) = 1 \) if \( a \in x \) and \( \varphi_a(x) = 0 \) if \( a \notin x \).
        \item Perhaps surprisingly \( H A \) has a countable orthogonal dual basis for every probabilistic coherence space \( A \).
        The point is that, for each \( a \in |A| \), \( \gamma_a \defe \sup \{ x(a) \mid x \in \PConfig(A) \} \) converges in \( \mathbb{R} \) and \( e_a \colon |A| \to \mathbb{R}_{\ge 0} \) defined by \( e_a(a) = \gamma_a \) and \( e_a(a') = 0 \) (\( a' \neq a \)) belongs to \( \PConfig(A) \).
        Then \( x \in \PConfig(A) \) can be canonically expressed as \( \sum_{a \in|A|} r_a \cdot e_a \), where \( r_a = x(a)/\gamma_a \).
        Hence \( (e_a, \varphi_a)_{a \in |A|} \) is a basis of \( HA \), where \( \varphi_a(x) = x(a)/\gamma_a \).
    \end{enumerate}
\end{example}

\begin{lemma}\label{lem:classic:vector-closure}
    \( \SVec[\srig] \) contains the zero module and the tensor unit \( \srig \) and is closed under \( \otimes \), \( \multimap \) and countable (co)products.
\end{lemma}
\begin{proof}
    The former claim is trivial.
    Let \( V,W \in \SVec[\srig] \) and \( (e_i, \varphi_i)_{i \in I} \) and \( (e'_j, \varphi'_j)_{j \in J} \) be orthogonal dual bases.
    For \( V \otimes W \), a basis is \( (e_i \otimes e'_j, \varphi_i \varphi'_j)_{(i,j) \in I \times J} \), where \( \varphi_i \varphi'_j \colon V \otimes W \longrightarrow \srig \) is the \( \srig \)-linear map corresponding to the bilinear map \( (\varphi_i \varphi'_j)(x, y) \defe \varphi_i(x) \varphi'_j(y) \).
For \( V \multimap W \), a basis is \( (e'_j \varphi_i, \varphi'_j({-}(e_i))_{(i,j) \in I \times J} \), where \( (\varphi'_j({-}(e_i)) \) maps \( f \in (V \multimap W) \) to \( \varphi'_j(f(e_i)) \).
    If \( V_i \) has a basis \( (e_{i,j}, \varphi_{i,j})_{j \in J_i} \) for every \( i \in I \), a basis for the product \( \prod_{i \in I} V_i \) is \( (e_{i,j}, \varphi_{i,j})_{i \in I, j \in J_i} \).
    The coproduct \( \coprod_{i \in I} V_i \) has the same basis.
\end{proof}

Lemma~\ref{lem:classic:vector-closure} provides us with a way to express an \( \srig \)-linear map \( f \colon V \longrightarrow W \) between modules with orthogonal dual bases as a matrix.
Let \( (e_i, \varphi_i)_{i \in I} \) and \( (d_j, \psi_j)_{j \in J} \) be bases of \( V \) and \( W \).
Then \( (d_j \varphi_i, \psi_j({-}(e_i)))_{(i,j) \in I \times J} \) is a basis of \( V \multimap W \).
So
\begin{equation*}
    f = \sum_{(i,j) \in I \times J} (\psi_j(f(e_i))) \cdot (d_j \varphi_i).
\end{equation*}
Let us write \( f = (f_{i,j})_{i \in I, j \in J} \) where \( f_{i,j} = \psi_j(f(e_i)) \).
An easy calculation shows that the matrix representation of a composite \( g \circ f \) is obtained by the matrix composition of matrix representations, i.e.~\( (g \circ f)_{i,k} = \sum_{j} g_{j,k} f_{i,j} \) for \( f = (f_{i,j})_{i \in I, j \in J} \) and \( g = (g_{j,k})_{j \in J, k \in K} \).
We emphasise here that a matrix representation depends on the choice of a basis; different choices may result in different matrices.
\begin{example}
    We use bases and symbols in Example~\ref{eg:module:vectors}.
    \begin{itemize}
        \item Recall that a linear map \( f \colon A \longrightarrow B \) between coherence spaces is a relation \( f \subseteq |A| \times |B| \) satisfying a certain property.
        By using bases given in Example~\ref{eg:module:vectors}(1),
        \( (Ff)_{a,b} = 1 \Leftrightarrow (a,b) \in f \).
\item Recall that a linear map \( f \colon A \longrightarrow B \) between probabilistic coherence spaces is a matrix \( f = (f_{a,b})_{a \in |A|, b \in |B|} \).
        Let us use the bases and symbols in Example~\ref{eg:module:vectors}(3).
        Then \( (Hf)_{a,b} = \varphi_b(Hf(e_a)) = (\gamma_a/\gamma_b) f_{a,b} \).
        Let us calculate the composition \( (Hg) \circ (Hf) \) using the matrix representation: \( ((Hg) \circ (Hf))_{a,c} = \sum_{b \in |B|} (Hg)_{b,c} (Hf)_{a,b} = \sum_{b \in |B|} ((\gamma_b/\gamma_c) g_{b,c}) \cdot ((\gamma_a/\gamma_b) f_{a,b}) = (\gamma_a/\gamma_c) \sum_{b} g_{b,c} f_{a,b} = (H(g \circ f))_{a,c} \) as required.
    \end{itemize}
\end{example}

Given \( V \in \SVec[\srig] \), the previous lemma shows that \( V^{\bot\bot} = (V \multimap \srig) \multimap \srig \) also belongs to \( \SVec[\srig] \).
Hence the double-negation monad \( ({-})^{\bot\bot} \) on \( \SMod[\srig] \) restricts to a monad on \( \SVec[\srig] \).
We show that it is idempotent: the proof is a direct calculation using a basis.
\begin{lemma}\label{lem:classic:vector-idempotent}
    The restriction of the double-negation monad \( ({-})^{\bot\bot} \colon \SMod[\srig] \longrightarrow \SMod[\srig] \) to \( \SVec[\srig] \) is idempotent.
\end{lemma}

Let \( \SVec[\srig]^{\bot\bot} \) be the full subcategory of \( \SVec[\srig] \) consisting of \( V \in \SVec[\srig] \) such that the canonical morphism \( \eta_V \colon V \longrightarrow V^{\bot\bot} \) is an isomorphism.
Since the double-negation monad \( ({-})^{\bot\bot} \) is idempotent, \( \SVec[\srig]^{\bot\bot} \) is a reflective subcategory of \( \SVec[\srig] \).
\begin{theorem}\label{thm:classic:vector}
    \( \SVec[\srig]^{\bot\bot} \) is a model of classical MALL (i.e.~a \( * \)-autonomous category with products), where the tensor product is given by \( (\smod \otimes \smodb)^{\bot\bot} \).
    If \( \SVec[\srig] \) is closed under the cofree cocommutative comonoid \( ! \) on \( \SMod[\srig] \), then \( \SVec[\srig]^{\bot\bot} \) has an linear exponential comonad whose underlying functor is \( (!\smod)^{\bot\bot} \). 
\end{theorem}
\begin{proof}
    For the former, the point is that \( (\smod \otimes \smodb)^{\bot\bot} \) is the representing object of bilinear maps.
    For the latter, we have
    \begin{equation*}
        \xymatrix{
            \Comon(\SVec[\srig])
            \ar@{}[r]|(0.6){\bot}
            \ar@<5pt>[r]
            & \SVec[\srig] \ar@<5pt>[l] \ar@<5pt>[r] \ar@{}[r]|(0.47){\bot} & \SVec[\srig]^{\bot\bot} \ar@<5pt>[l]
        }
    \end{equation*}
    and left adjoints are strong monoidal.
\end{proof}

We give a sufficient condition for the closure of \( \SVec[\srig] \) by the cofree cocommutative comonoid \( ! \) of \( \SMod[\srig] \).
An \emph{ideal} of \( \srig \) is a submodule of \( \srig \in \SMod[\srig]\); a \emph{downward-closed sum-reflecting ideal} is a downward-closed sum-reflecting submodule of \( \srig \).
\begin{theorem}\label{thm:classical:vector:exponential}
    Suppose that every downward-closed and sum-reflecting ideal of \( \srig \) belongs to \( \SVec[\srig] \).
    Then \( \SVec[\srig] \) is closed under the cofree cocommutative comonoid \( ! \) of \( \SMod[\srig] \).
\end{theorem}
\begin{proof}
    The idea is to realise \( !V \) as a submodule \( \S V \) of \( \prod_n \SymTensor{V}{n} \), where \( \SymTensor{V}{n} \) is the \( n \)-th symmetric tensor power, i.e.~the equaliser of the \( n! \) permutations of \( V^{\otimes n} \longrightarrow V^{\otimes n} \).
    Intuitively the restriction of a map \( \prod_{n} V^{\otimes n} \longrightarrow (\prod_n V^{\otimes n}) \boxtimes (\prod_n V^{\otimes n}) \) (where \( \boxtimes \) is the tensor product \( \otimes \); we use symbol different from \( \otimes \) in order to avoid confusion) given by
    \begin{equation*}
        x_1 \otimes \dots \otimes x_n \mapsto \sum_{0 \le i \le n} (x_1 \otimes \dots \otimes x_i) \boxtimes (x_{i+1} \otimes \dots x_n)
    \end{equation*}
    to \( \S V \longrightarrow \S V \otimes \S V \)
should be the comultiplication.
    \ifwithappendix
    The main concern is the definedness: see Appendix~\ref{sec:appx:classic:exponential-proof} for details.
    \else
    The main concern is the definedness.
    \fi
\end{proof}

\begin{example}
    Here we just list some results.
    The equivalence below means the equivalence as \( * \)-autonomous categories.
    \begin{itemize}
        \item \( \SVec[\CohRing]^{\bot\bot} \) is equivalent to \( \CCoh \), the category of countable coherence spaces.  Since downward-closed sum-reflecting ideals of \( \CohRing \) are \( \{ 0,1 \} \) and \( \{ 0 \} \), both of which belong to \( \SVec[\CohRing] \), \( \SVec[\CohRing] \) is closed under \( ! \).
        The induced exponential comonad on \( \SVec[\CohRing]^{\bot\bot} \) is the mulitset exponential.
        \item \( \SVec[\FinRing]^{\bot\bot} \) is equivalent to \( \CFin \), the category of countable finiteness spaces.  
        By the same reason as above, \( \SVec[\FinRing] \) is closed under \( ! \), so it is a model of classical linear logic.
\item \( \SVec[\BoolRing]^{\bot\bot} \) is equivalent to the category of countable sets and relations.  \( \SVec[\BoolRing] \) is closed under \( ! \) and the induced exponential comonad on \( \SVec[\BoolRing]^{\bot\bot} \) is the finite multiset comonad.
        \item \( \SVec[{[0,1]}]^{\bot\bot} \) is equivalent to \( \CPCoh \).  Unfortunately \( \SVec[{[0,1]}] \) is not closed under \( ! \) of \( \SMod[{[0,1]}] \).  See the next subsection for a way to reconstruct \( ! \) on \( \CPCoh \).
        \item Let \( \srig \) be a complete \( \Sigma \)-semiring.  Then \( \SVec[\srig]^{\bot\bot} \) is equivalent to the weighted relational model \( \CWRel_\srig \) \cite{Laird2013,Laird2016} over \( \srig \).
        We can show \( \SVec[\srig] \) is closed under \( ! \) and hence \( \SVec[\srig]^{\bot\bot} \) has a linear exponential comonad, although the assumption of Theorem~\ref{thm:classical:vector:exponential} does not hold for complete \( \Sigma \)-semiring \( \srig \) in general.
\end{itemize}
\end{example}

\subsection{Modules with countable dual basis}
The following definition gives another well-behaved class of modules.
The definition comes from a characterisation of finitely generated projective module, although in our \( \Sigma \)-semiring setting these modules are no longer projective in the categorical sense.
\begin{definition}[Module with countable dual basis]
    Let \( \srig \) be a \( \Sigma \)-semiring and \( \smod \) be an \( \srig \)-module.
    A family \( (e_i, \varphi_i)_{i \in I} \) of pairs of \( e_i \in |\smod| \) and \( \varphi_i \colon \smod \longrightarrow \srig \) is a \emph{dual basis} if \( x = \sum_{i \in I} \varphi_i(x) \cdot e_i \) for every \( x \in |\smod| \).
    A dual basis is \emph{countable} if \( I \) is countable.
    We write \( \DBMod[\srig] \) for the full subcategory of \( \SMod[\srig] \) consisting of those having countable dual bases.
\end{definition}
Thus the only difference between an orthogonal dual basis and a dual basis is whether \( \varphi_i(e_j) = \delta_{i,j} \).

\begin{proposition}
    Let \( \smod \) be an \( \srig \)-module.
    It has a countable dual basis, i.e.~\( \smod \in \DBMod[\srig] \), if and only if there is a retraction \( \smod \lhd V \) in \( \SMod[\srig] \) for some module \( V \in \SVec[\srig] \) having a countable orthogonal dual basis.
    In other words, \( \DBMod[\srig] \) is equivalent to the idempotent completion (or Karoubi envelop) of \( \SVec[\srig] \).
\end{proposition}

By an argument similar to the case of \( \SVec[\srig] \), one can prove that \( \DBMod[\srig] \) is closed under \( \otimes \), \( \multimap \) and countable (co)products.
\begin{lemma}\label{lem:classic:proj-closure}
    The zero module and \( \srig \) have countable dual bases. 
    \( \srig \)-modules with countable dual bases are closed under \( \otimes \), \( \multimap \) and countable (co)products.
\end{lemma}

\begin{lemma}\label{lem:classic:proj-idempotent}
    The restriction of the double-negation monad \( ({-})^{\bot\bot} \) to \( \DBMod[\srig] \) is idempotent.
\end{lemma}

\( \DBMod[\srig]^{\bot\bot} \) is the full subcategory of \( \DBMod[\srig] \) consisting of \( \srig \)-modules \( \smod \in \DBMod[\srig] \) such that the canonical morphism \( \smod \longrightarrow \smod^{\bot\bot} \) is an isomorphism.
\begin{theorem}
    \( \DBMod[\srig]^{\bot\bot} \) is a model of classical MALL, where the tensor product is given by \( (\smod \otimes \smodb)^{\bot\bot} \).
    If\/ \( \DBMod[\srig] \) is closed under the cofree cocommutative comonoid \( ! \) on \( \SMod[\srig] \), then \( \DBMod[\srig]^{\bot\bot} \) has an linear exponential comonad whose underlying functor is \( (!\smod)^{\bot\bot} \). 
\end{theorem}
\begin{proof}
    Similar to the proof of Theorem~\ref{thm:classic:vector}, but use Lemmas~\ref{lem:classic:proj-closure} and \ref{lem:classic:proj-idempotent} instead of Lemmas~\ref{lem:classic:vector-closure} and \ref{lem:classic:vector-idempotent}.
\end{proof}

\begin{theorem}\label{thm:classical:proj:exponential}
    Suppose that every downward-closed sum-reflecting ideal of \( \srig \) is in \( \DBMod[\srig] \).
    Then \( \DBMod[\srig] \) is closed under \( ! \) of \( \SMod[\srig] \).
\end{theorem}
\begin{proof}
    By the same argument as the proof of Theorem~\ref{thm:classical:vector:exponential}.
\end{proof}

\begin{example}\label{eg:classical:proj}
    Here we just list some results.
    The equivalence below means the equivalence as \( * \)-autonomous categories.
    \begin{itemize}
        \item \( \DBMod[\CohRing]^{\bot\bot} \) is equivalent to \( \CCoh \), the category of countable coherence spaces.  \( \DBMod[\CohRing] \) is closed under \( ! \) and the induced exponential comonad on \( \DBMod[\CohRing]^{\bot\bot} \) is the mulitset exponential.
        \item \( \DBMod[{[0,1]}]^{\bot\bot} \) is a model of classical logic.
        The condition of Theorem~\ref{thm:classical:proj:exponential} can be checked as follows.
        Let \( A \) be a downward-closed subset of \( [0,1] \).
        Then either \( A = [0,\gamma] \) or \( [0,\gamma) \) for some \( \gamma \in [0,1] \).
        The former has the singleton family \( (\gamma, (x \mapsto x/\gamma)) \) as a basis.
        For the latter, let \( r \in [0,1] \) be a real number such that \( r < \gamma < 2r \).
        Then \( (e_i, \varphi_i)_{i = 1,2} \) given by \( e_1 = e_2 = r \) and \( \varphi_1(x) = \varphi_2(x) = x/2r \) is a dual basis, since \( \varphi_1(x) \cdot e_1 + \varphi_2(x) \cdot e_2 = (x/2r) \cdot r + (x/2r) \cdot r = x \). 
        Hence every downward-closed sum-reflecting ideal of \( [0,1] \) has a dual basis.
        One can prove that \( \DBMod[{[0,1]}] \) is equivalent to \( \CPCoh \),
        \ifwithappendix
        although the proof is not trivial (see Appendix~\ref{sec:appx:pcoh}).
        \else
        although the proof is not trivial.
        \fi
        The induced exponenital coincides with the standard one, which has been proved to be free~\cite{Abou-Saleh2013}.
        \item \( \DBMod[\BoolRing]^{\bot\bot} \) contains the Scott model of linear logic~\cite{Ehrhard2012a,Winskel2004} as a submodel.
        The category \( \DBMod[\BoolRing]^{\bot\bot} \) can be characterised as the category of \emph{prime-continuous complete lattices with a countable basis}~(cf.~\cite[Definition~4.3]{Negri2002})
and join-preserving maps.
    \end{itemize}
\end{example}

\section{Gluing and Orthogonality}\label{sec:gluing}
This section briefly discusses another general approach for constructing models of linear logic, namely the approach based on \emph{gluing} and \emph{orthogonality} by Hyland and Schalk~\cite{Hyland2003}.
After briefly reviewing the framework, we shall see that an appropriate choice of the coefficient ring is beneficial in their approach as well.

\paragraph{Gluing and orthogonality in Hyland and Schalk}
We briefly review their method, focussing on a way to produce the coherence space model.

Let \( \cat \) be a categorical model of classical linear logic, i.e.~a \( * \)-autonomous category with products and linear exponential comonad.
The category \( G(\cat) \) obtained by \emph{double gluing along hom-functors} is described as follows: An object \( A = (R, U, X) \) of \( G(\cat) \) is given by an object \( R \in \cat \) and subsets \( U \subseteq \cat(I, R) \) and \( X \subseteq \cat(R, \neg I) \), where \( I \) is the tensor unit;
A morphism from \( A = (R, U, X) \) to \( B = (S, V, Y) \) is a morphism \( f \colon R \longrightarrow S \) in \( \cat \) such that \( I \stackrel{u}{\longrightarrow} R \stackrel{f}{\longrightarrow} S \) is in \( V \) for every \( u \in U \) and \( R \stackrel{f}{\longrightarrow} S \stackrel{v}{\longrightarrow} \neg I \) is in \( X \) for every \( v \in V \).
\( G(\cat) \) is a \( * \)-autonomous category with products.
Furthermore, if the exponential is cofree, \( G(\cat) \) has a linear exponential comonad as well.

An \emph{orthogonality} is a family of relations between morphisms in \( \cat \), used to carve out a full subcategory of \( G(\cat) \).
It is a family \( \bot = (\bot_R)_{R \in \cat} \) of relations, where \( ({\bot_R}) \subseteq \cat(I, R) \times \cat(R, \neg I) \) for each \( R \in \cat \), that satisfies certain conditions.
Given an orthogonality \( \bot \) and \( U \subseteq \cat(I, R) \), its \emph{orthogonal} \( U^\circ \) is defined as \( \{ x \in \cat(R, \neg I) \mid \forall u \in U. u \bot_R x \} \).
Similarly the orthogonal \( X^\circ \) of \( X \subseteq \cat(R, \neg I) \) is \( \{ u \in \cat(I, R) \mid \forall x \in X. u \bot_R x \} \).
The \emph{tight orthogonality subcategory} \( T(\cat) \) is the full subcategory of \( G(\cat) \) consisting of objects \( (R, U, X) \) such that \( U = X^\circ \) and \( X = U^\circ \).

A reader familiar with models of linear logic such as coherence spaces~\cite{Girard1987}, totality spaces~\cite{Loader1994}, finiteness spaces~\cite{Ehrhard2005a} and probabilistic coherence spaces~\cite{Girard2004,Danos2011} should easily recognise that the construction of tight orthogonality subcategory captures the core idea of these models.
Hyland and Schalk proved that, if the orthogonality \( \bot \) satisfies a number of conditions and the exponential of \( \cat \) is cofree, \( T(\cat) \) is a model of classical linear logic~\cite[Proposition~64]{Hyland2003}.

Hyland and Schalk discuss a way to carve out the category \( \CCoh \) of coherence spaces from \( G(\CRel) \)~\cite[Example~66(3)]{Hyland2003}.
Since the tensor unit \( I \) of \( \CRel \) is a singleton set and \( \neg I = I \), both \( u \in \CRel(I, R) \) and \( x \in \CRel(R, \neg I) \) can be identified with subsets of \( R \).
The orthogonality is defined as expected: \( u \bot_R x \) if and only if \( \#(u \cap x) \le 1 \).
However here is a subtlety: this orthogonality does not satisfy the condition of \cite[Proposition~64]{Hyland2003} and hence their general construction cannot be applied to this case.
A modification to deal with this situation is discussed in \cite[Example~66(3)]{Hyland2003}.

\paragraph{Significance of the coefficient $\Sigma$-semiring}
Let us analyse the issue from the viewpoint of the coefficient \( \Sigma \)-semiring.

Recall that \( \CRel \) of sets and relations is exactly the category of weighted relations over the boolean \( \Sigma \)-semiring \( \rbool \) and that \( \CCoh \) of coherence spaces is a subcategory of the category of \( \rcoh \)-modules.
Although the underlying sets of \( \rbool \) and \( \rcoh \) coincide, their algebraic structures differ significantly.
This mismatch causes the subtlety mentioned above.

Let \( \CWRel_{\Nat_\infty} \) be the category of weighted relations over \( \Nat_\infty \), where \( \Nat_\infty = \Nat \cup \{ \infty \} \) is the completion of \( \rcoh \).
In this category, the tensor unit and its dual are a singleton set and \( u \in \CWRel_{\Nat_\infty}(\{\ast\}, R) \) and \( x \in \CWRel_{\Nat_\infty}(R, \{\ast\}) \) can be identified with functions \( u, x \colon R \longrightarrow \Nat_\infty \).
Consider the orthogonality relation \( \bot \) defined by \( u \bot_R x \) if and only if \( \sum_{e \in R} u(e) \cdot x(e) \le 1 \), which coincide with the condition \( (x \circ u) \le 1 \) where \( x \circ u \) is the composition in \( \CWRel_{\Nat_\infty} \) and the \( (1\times 1) \)-matrix \( (x \circ u) \colon \{\ast\} \longrightarrow \{\ast\} \) is identified with \( \Nat_\infty \). 
This orthogonality is an example of particularly well-behaved orthogonality, called \emph{focused orthogonality}, which automatically satisfies the conditions of \cite[Proposition~64]{Hyland2003}.
Hence \( T(\CWRel_{\Nat_\infty}) \) is a model of classical linear logic with exponential, from which one can extract a category equivalent to \( \CCoh \).\footnote{This last step, extracting an appropriate reflective subcategory, is actually required in the reconstruction of \( \CCoh \) by Hyland and Schalk as well.  See a comment on the difference between the totality space model \( \mathbf{Tot} \) and \( T(\CRel) \) in \cite[Example~66(2)]{Hyland2003}; although it is not mentioned, the same difference can be found in the reconstruction of \( \CCoh \) in \cite[Example~66(3)]{Hyland2003}.}

 \section{Conclusion and Future Work}
This paper develops a general framework, parameterised by \( \Sigma \)-semiring \( \srig \), for constructing models of intuitionistic and classical linear logic.
In particular, we saw how the choice of the coefficient \( \Sigma \)-semiring \( \srig \) affects the structure of models, by using \( \srig = \CohRing, \BoolRing, \FinRing \) and \( [0,1] \) as examples.

One may be interested in other kinds of partial algebras.
For example, an axiom saying that \( \sum_{i \in I} x_i \) is defined if \( \sum_{i \in J} x_i \) is defined for every finite subset \( J \subseteq I \) is studied in \cite{Hines2007,Manes1986}.
For another example, one may be interested in dropping an axiom of \( \Sigma \)-module
in order to reason about K\"othe sequence spaces~\cite{Ehrhard2002a} (see Example~\ref{ex:monoid:Kothe}).
Since these axioms are partial Horn theories, once an SMCC structure is established, they become a Lafont model of intuitionistic linear logic.

Another interesting topic is the categorification of the module theory of this paper.
The category \( \mathbf{CPM}_{\mathbf{s}} \)~\cite{Selinger2004} is a model of quantum linear programs, which has an additive structure in each hom-set and a multiplicative structure given by the composition of morphisms.
This additive structure has been used to construct a model of quantum programs with recursions and infinite data~\cite{Pagani2014,Tsukada2018}.
What does the category of modules look like if we choose the \( \SMon \)-enriched SMCC \( \mathbf{CPM}_{\mathbf{s}} \) as the categorified \( \Sigma \)-semiring?

\bibliographystyle{ACM-Reference-Format}
\bibliography{library.bib}

\clearpage
\appendix
\section{Supplementary materials for Section~\ref{sec:partial}}

\subsection{Detailed Proof of Theorem~\ref{thm:pre:horn-presentable}}
Let us first recall essentially algebraic theories (see, e.g., \cite{Adamek1994}).
\begin{definition}[Essentially algebraic theory]
    Let \( S \) be a set of sorts and \( \Sigma \) be an \( S \)-sorted signature of \( \kappa \)-ary algebra.
    An \emph{\( \kappa \)-ary essentially algebraic theory} is a subclass of \( \kappa \)-ary partial Horn theory with additional information:
    \begin{itemize}
        \item Operations in \( \Sigma \) is divided into two classes: total operations \( \Sigma_t \) and partial operations \( \Sigma_p \).
        \item Each total operation \( (\sigma \colon \prod_{i < \alpha} s_i \longrightarrow s) \in \Sigma_t \) must be total, i.e.,
            \begin{equation*}
                \left(
                    \forall (x_i \colon s_i)_{i < \alpha}.~~~~\IsDef{\sigma(x_i)_i}
                \right)
                \qquad\in \Th.
            \end{equation*}
        \item Each partial operation \( (\sigma \colon \prod_{i < \alpha} s_i \longrightarrow s) \in \Sigma_p \) is associated with a formula of the form
            \begin{equation*}
                \left(
                    \forall (x_k \colon s'_k)_{k < \gamma}.~\left(\bigwedge_{j < \beta} t_j = t'_j \right) \Longleftrightarrow \IsDef{\sigma(y_i)_i}
                \right)
                \qquad\in \Th,
            \end{equation*}
            where \( t_j \) and \( t'_j \) contain only total operations and \( \beta,\gamma < \kappa \).
            This formula characterises the domain of \( \sigma \).
        \item All other formulas in the theory are of the form
            \begin{equation*}
                \left(
                    \forall (x_i \colon s_i)_{i < \alpha}.~ \IsDef{t} \wedge \IsDef{t'} \Longrightarrow t = t'
                \right)
                \qquad\in \Th,
            \end{equation*}
            where \( \alpha < \kappa \) and \( t \) and \( t' \) terms possibly containing partial operations.
    \end{itemize}
    The last condition can be relaxed: one can use arbitrary (\( \kappa \)-ary) implication~\cite[3.35(3)]{Adamek1994}.
\end{definition}

Let \( S \) be a set of sort, \( \Sigma \) be an \( S \)-sorted signature of \( \kappa \)-ary algebra and \( \Th \) be a \( \kappa \)-ary partial Horn theory.
We first transform \( \Th \) into a certain form.

A \emph{basic equation} is of the form \( x = \sigma(y_i)_i \).
A Horn formula is \emph{basic} if it is of the form \( \forall (x_i \colon s_i)_{i < \alpha}. \bigwedge_{j < \beta} \psi_j \Longrightarrow \psi \) where each \( \psi_j \) is a basic equation and \( \psi \) is an equation \( y = z \) over variables or \( \IsDef{\sigma(y_i)_i} \).
Every Horn formula can be translated into a set of basic formulas without changing its meaning as follows.
\begin{itemize}
    \item A formula
    \begin{equation*}
        \forall \overrightarrow{x \colon s}. \bigwedge_{j} \psi_j \Longrightarrow t = t'
    \end{equation*}
    is transformed into the set of
    \begin{align*}
        \forall \overrightarrow{x \colon s}. \bigwedge_{j} \psi_j &\Longrightarrow \IsDef{t} \\
        \forall \overrightarrow{x \colon s}. \bigwedge_{j} \psi_j &\Longrightarrow \IsDef{t'} \\
        \forall \overrightarrow{x \colon s}. \forall (y \colon s). \forall (z \colon s). \bigwedge_{j} \psi_j &\wedge y = t \wedge z = t' \Longrightarrow y = z
    \end{align*}
    where \( s \) is the sort of \( t \) and \( t' \).
    \item A formula
    \begin{equation*}
        \forall \overrightarrow{x \colon s}. \bigwedge_{j} \psi_j \Longrightarrow \IsDef{\sigma(t_i)_i}
    \end{equation*}
    (\( t_i \) is not a variable for some \( i \)) is transformed into the set of formulas consisting of
    \begin{align*}
        \forall \overrightarrow{x \colon s}. \bigwedge_{j} \psi_j &\Longrightarrow \IsDef{t_i}
    \end{align*}
    for each \( i \) and
    \begin{align*}
        \forall \overrightarrow{x \colon s}. \forall (y_i \colon s_i)_i. \bigwedge_{j} \psi_j &\wedge \bigwedge_i y_i = t_i \Longrightarrow \IsDef{\sigma(y_i)_i}
    \end{align*}
    where \( \sigma \colon \prod_{i < \alpha} s_i \longrightarrow s \).
    \item A formula
    \begin{equation*}
        \forall \overrightarrow{x \colon s}. t = t' \wedge \bigwedge_{j} \psi_j \Longrightarrow \psi
    \end{equation*}
    (\( t \) is not a variable) is transformed into
    \begin{align*}
        \forall \overrightarrow{x \colon s}. \forall (y \colon s). y = t \wedge y = t' \wedge \bigwedge_{j} \psi_j &\Longrightarrow \psi
    \end{align*}
    where \( s \) is the sort of \( t \) and \( t' \).
    \item A formula
    \begin{equation*}
        \forall \overrightarrow{x \colon s}. z = \sigma(t_i)_i \wedge \bigwedge_{j} \psi_j \Longrightarrow \psi
    \end{equation*}
    (\( t_i \) is not a variable for some \( i \)) is transformed into
    \begin{align*}
        \forall \overrightarrow{x \colon s}. \forall (y_i \colon s_i)_i. \bigwedge_i y_i = t_i \wedge z = \sigma(y_i)_i \wedge \bigwedge_{j} \psi_j &\Longrightarrow \psi
    \end{align*}
    where \( s \) is the sort of \( t \) and \( t' \).
\end{itemize}
For each Horn formula, this rewriting process terminates at most \( \kappa \) steps because \( \kappa \) is a regular cardinal and \( \Sigma \) and \( \Th \) are \( \kappa \)-ary.
The result of the rewriting is again a \( \kappa \)-ary Horn theory (by the regularity of \( \kappa \)).

We assume that \( \Th \) is basic.
We construct a set \( S' \) of sorts, an \( S' \)-sorted signature \( \Sigma' \) and an essentially algebraic theory \( \Th' \).
\begin{itemize}
    \item The set \( S' \) of sorts is defined by
    \begin{equation*}
        S' \;\;\defe\;\; S \uplus \{ \mathrm{dom}_\sigma \mid \sigma \in \Sigma \}.
    \end{equation*}
    \item The set \( \Sigma' \) of operations consists of
    \begin{itemize}
        \item for each \( (\sigma \colon \prod_{i < \alpha} s_i \longrightarrow s) \in \Sigma \), the operation
        \begin{equation*}
            \hat{\sigma} \colon \mathrm{dom}_\sigma \longrightarrow s
        \end{equation*}
        and projections
        \begin{equation*}
            \pi^{\sigma}_i \colon \mathrm{dom}_\sigma \longrightarrow s_i
        \end{equation*}
        for each \( i < \alpha \), and
        \item for each formula \( \psi \in \Th \) whose conclusion is \( \IsDef{\sigma(x_i)_i} \) and \( (\sigma \colon \prod_{i < \alpha} s_i \longrightarrow s) \in \Sigma \), the tupling
        \begin{equation*}
            \tau_\psi \colon \prod_{i < \alpha} s_i \longrightarrow \mathrm{dom}_\sigma.
        \end{equation*}
    \end{itemize}
    \item \( \hat{\sigma} \) and \( \pi^\sigma_i \) are total and \( \tau_\psi \) is partial.
    \item If
    \begin{equation*}
        \psi \;\;=\;\;
        (\forall \overrightarrow{x \colon s}. \bigwedge_i y_i = \sigma_i(y_{i,j})_{j < \alpha_i} \Longrightarrow \IsDef{\sigma(z_k)_k}),
    \end{equation*}
    then the formula characterising the domain of \( \tau_\psi \) is
    \begin{align*}
        &
        \forall \overrightarrow{x \colon s}. \forall (d_i \colon \mathrm{dom}_{\sigma_i})_i.
        \\ & \qquad
        \bigwedge_i \Big( y_i = \hat{\sigma}_i(d_i) \wedge \bigwedge_{j < \alpha_i} y_{i,j} = \pi^{\sigma_i}_j(d_i) \Big)
        \\ & \qquad
        \Longleftrightarrow \tau_\psi(z_k)_k.
    \end{align*}
\end{itemize}
The theory \( \Th' \) contains the following formulas, in addition to the above-defined formulas defining the domains of \( \tau_\psi \).
\begin{itemize}
    \item For each \( (\sigma \colon \prod_{i < \alpha} s_i \longrightarrow s) \in \Sigma \), the formula
    \begin{equation*}
        \forall (d, d' \colon \mathrm{dom}_\sigma). \bigwedge_{i < \alpha} \pi^{\sigma}_i(d) = \pi^{\sigma}_i(d') \Longrightarrow d = d',
    \end{equation*}
    meaning that \( \mathrm{dom}_\sigma \) must be the subset of \( \prod_{i < \alpha} s_i \).
    \item For each \( \psi \in \Th \) concluding \( \IsDef{\sigma(x_i)_i} \) (the sort of \( \sigma \in \Sigma \) is \( \prod_{i < \alpha} s_i \longrightarrow s \)), the formula
    \begin{equation*}
        \forall (x_i \colon s_i)_{i < \alpha}. \IsDef{\tau_\psi(x_i)_i} \Longrightarrow \pi^{\sigma}_j(\tau_\psi(x_i)_i) = x_j
    \end{equation*}
    for each \( j < \alpha \).
    This means that \( \tau_\psi(x_i)_i \) is actually the tuple \( (x_i)_i \) provided that it is defined.
    \item For each \( \psi \in \Th \) concluding \( z = z' \), i.e.,
    \begin{equation*}
        \psi
        \;\;=\;\;
        \left(
            \forall \overrightarrow{x \colon s}. \bigwedge_i y_i = \sigma_i(y_{i,j})_{j < \alpha_i} \Longrightarrow z = z',
        \right)
    \end{equation*}
    the formula
    \begin{align*}
        &
        \forall \overrightarrow{x \colon s}. \forall (d_i \colon \mathrm{dom}_{\sigma_i})_i.
        \\ & \qquad
        \bigwedge_i \Big( y_i = \hat{\sigma}_i(d_i) \wedge \bigwedge_{j < \alpha_i} y_{i,j} = \pi^{\sigma_i}_j(d_i) \Big)
        \\ & \qquad
        \Longrightarrow z = z'.
    \end{align*}
\end{itemize}

It is not difficult to see that the category of models of \( \Th \) is equivalent to the category of models of \( \Th' \).
A model \( ((A_s)_{s \in S}, (\sigma_A)_{\sigma \in \Sigma}) \) of \( \Th \) induces a model of \( \Th' \) such that \( A_{\mathrm{dom}_\sigma} = \{ (a_i)_i \in \prod_i s_i \mid \IsDef{\sigma_A(a_i)_i} \} \), \( \pi^\sigma_j(a_i)_i = a_i \) and \( \tau_\psi \) is the identity (on the domain specified by the theory).
A homomorphism \( f \colon A \longrightarrow B \) of \( \Th \)-algebras canonically defines a homomorphism \( A' \longrightarrow B' \) of the induced \( \Th' \)-algebras, since the action of the homomorphism on \( \mathrm{dom}_\sigma \) is uniquely determined by the interpretations of \( \pi^\sigma_i \) in \( A' \) and \( B' \).
This gives us a fully faithful functor from models of \( \Th \) to models of \( \Th' \).
This is essentially surjective on objects since, for every model \( A \) of \( \Th' \), the interpretation \( A_{\mathrm{dom}_\sigma} \) of \( \mathrm{dom}_\sigma \) can be identified with a subset of \( \prod_i A_{s_i} \) (where \( \sigma \colon \prod_i s_i \longrightarrow s \) in \( \Sigma \)).

\subsection{Proof of Lemma~\ref{lem:partial:presentation}}
Consider the signature \( \Sigma' \defe \Sigma \uplus \biguplus_{s \in S} G_s \), where \( x \in G_s \) is seen as an nullary operation of sort \( s \), and the theory \( \Th' \defe \Th \cup \mathcal{R} \cup \{ \IsDefined{g} \mid g \in \bigcup_{s \in S} G_s \} \), where \( (t,t') \in \bigcup_{s \in S} R_s \) is regarded as the formula \( t = t' \).
Then a partial algebra \( B' \in \PAlg(\Sigma', \Th') \) consists of a partial algebra \( B \in \PAlg(\Sigma,\Th) \) together with assignments \( \vartheta_s \colon G_s \longrightarrow A_s \) given by \( g \mapsto \sem{g}_{B'} \); note that \( \sem{g}_{B'} \) is defined since \( \IsDefined{g} \) is in \( \Th' \).
Since \( (\Sigma', \Th') \) is a Horn theory, the category \( \PAlg(\Sigma', \Th') \) is locally presentable.
So it has the initial model \( A' = (A, \varrho) \), which is presented by \( (G,R) \).

\subsection{Proof of Lemma~\ref{lem:partial:generated}}
Consider a family of subsets \( B_s \subseteq A_s \) defined by
\begin{equation*}
    B_s \defe \{ t \mid t \in \Term^s_\Sigma(G), \IsDefined{t} \}.
\end{equation*}
The restriction of \( \sigma_A \) to \( (B_s)_{s \in S} \) defines a partial algebra \( B \) in \( \PAlg(\Sigma, \Th) \).
Furthermore \( g \in B_s \) for every \( g \in G_s \) and hence \( B \) is equipped with an assignment satisfying \( R \).
Hence there is a unique homomorphism \( h \colon A \longrightarrow B \).
Since \( A \longrightarrow B \hookrightarrow A \) is identity on \( G \), by the universal property of \( A = \langle G, R \rangle \), the composite is the identity.
So \( B \hookrightarrow A \) is actually an isomorphism, which must be a bijection on each sort, since a homomorphism is a family of functions on underlying sets and the identity consists of identity functions. \section{Supplementary materials for Section~\ref{sec:module}}

\subsection{Proof of Lemma~\ref{lem:module:tensor-exists}}
    We give a presentation of a representing object.
    The set \( \mathcal{G} \) of generators is \( \{ x \otimes y \mid (x,y) \in |\smod| \times |\smodb| \} \).
    The relation \( \mathcal{R} \) consists of \( x \otimes y = \sum_i r_i \cdot (x_i \otimes y) \) (where \( x = \sum_i r_i \cdot x_i \) in \( \smod \)) and \( x \otimes y = \sum_j r_j \cdot (x \otimes y_j) \) (where \( y = \sum_j r_j \cdot y_j \) in \( \smodb \)).

    We give a bijection \( \varrho_{\smodc} \colon \SMod[\srig](\langle \mathcal{G}, \mathcal{R} \rangle, \smodc) \cong \Bilin_{\srig}(\smod,\smodb; \smodc) \).
    Given an \( \srig \)-linear map \( f \colon \langle \mathcal{G}, \mathcal{R} \rangle \longrightarrow \smodc \), let \( \varrho_{\smodc}(f)(x,y) \defe f(x \otimes y) \).
    Then \( \varrho_{\smodc}(f) \) is \( \srig \)-bilinear: if \( \sum_i r_i \cdot x_i \) is defined, then \( \varrho(f)(\sum_i r_i \cdot x_i, y) = f((\sum_i r_i \cdot x_i) \otimes y) = f(\sum_i r_i \cdot (x_i \otimes y)) = \sum_i r_i \cdot f(x_i \otimes y) = \sum_i r_i \cdot \varrho(f)(x_i, y) \). 
    Conversely, given a \( \srig \)-bilinear map \( g \in \Bilin_{\srig}(\smod,\smodb; \smodc) \), the mapping \( \mathcal{G} \longrightarrow |\smodc|, x \otimes y \mapsto g(x,y) \) satisfies all conditions in \( \mathcal{R} \).
    By the universal property of \( \langle \mathcal{G}, \mathcal{R} \rangle \), there is a unique \( \srig \)-linear map \( f \colon \langle \mathcal{G}, \mathcal{R} \rangle \longrightarrow \smodc \) such that \( f(x \otimes y) = g(x,y) \).
    Then we define \( \varrho_{\smodc}^{-1}(g) \defe f \).

    It is easy to see that \( (\varrho_\smodc)_{\smodc} \) is a natural isomorphism.

\subsection{Proof of Lemma~\ref{lem:module:tensor-peak}}
Let \( \smodc \) be the full submodule of \( \smod \otimes \smodb \) consisting of \( \{ z \in |\smod \otimes \smodb| \mid \exists (x,y) \in |\smod|\times|\smodb|. z \le x \otimes y \} \).
The inclusion \( |\smodc| \hookrightarrow |\smod \otimes \smodb| \) is an \( \srig \)-linear map, which is a monomorphism.
Let \( \varrho \) be an assignment \( \{ x \otimes y \mid (x,y) \in |\smod|\times|\smodb| \} \longrightarrow |\smodc| \) defined by \( f(x \otimes y) \defe x \otimes y \).
It is easy to see that \( \varrho \) satisfies the relation in the presentation of \( \smod \otimes \smodb \) (see the proof of Lemma~\ref{lem:module:tensor-exists}).
Hence \( \rho \) can be extended to a \( \srig \)-linear map \( \smod \otimes \smodb \longrightarrow \smodc \).
Since the composite \( \smod \otimes \smodb \longrightarrow \smodc \hookrightarrow \smod \otimes \smodb \) is identity on generators, it is the identity.
So \( \smodc \hookrightarrow \smod \otimes \smodb \) is an isomorphism.

\subsection{Proof of Lemma~\ref{lem:module:injection}}
Since the composition is the function composition of underlying functions, every injection is a monomorphism.

Assume that \( f \colon \smod \longrightarrow \smodb \) is not an injection.
This means that \( f(x) = f(y) \) for some \( x,y \in |\smod| \) with \( x \neq y \).
Consider the map \( x^{\dagger} \colon \srig \longrightarrow \smod \) defined by \( x^{\dagger}(r) \defe r \cdot x \).
By the bilinearity of the action, \( x^{\dagger} \) is \( \srig \)-linear.
By the assumption,
\begin{equation*}
    (f \circ x^{\dagger})(r)
    =
    f(r \cdot x)
    =
    r \cdot f(x)
    =
    r \cdot f(y)
    =
    f(r \cdot y)
    =
    (f \circ y^{\dagger})(r)
\end{equation*}
for every \( r \).
Hence \( f \circ x^\dagger = f \circ y^{\dagger} \) but \( x^{\dagger}(1) = x \neq y = y^\dagger(1) \).
That means, \( f \) is not a monomorphism.

\subsection{Proof of Lemma~\ref{lem:module:equaliser}}
Let \( f,g \colon \smod \longrightarrow \smodb \) be \( \srig \)-linear maps.
Let \( \smod' \) be the full submodule consisting of \( \{ x \in |\smod| \mid f(x)=g(x) \} \).
Let \( \iota \) be the inclusion function.

We first prove that \( \smod' \) is a \( \srig \)-module.
Let \( x \in |\smod| \) and assume \( f(x) = g(x) \).
Then, for every \( r \in |\srig| \), we have \( f(r \cdot x) = r \cdot f(x) = r \cdot g(x) = g(r \cdot x) \).
Hence \( r \cdot x \in |\smod'| \).

A notable property of \( \smod' \) is as follows: for every family \( (x_i)_{i \in I} \) of \( |\smod'| \), if \( \sum_i x_i \) is defined in \( \smod \), then the sum \( \sum_i x_i \) is defined in \( \smod' \).
Assume that \( \sum_i x_i \) is defined in \( \smod \) for a family \( (x_i)_i \) of elements in \( |\smod'| \).
Then \( f(\sum_i x_i) \Kle \sum_i f(x_i) \) and \( g(\sum_i x_i) \Kle \sum_i g(x_i) \).
Since \( f(x_i) = g(x_i) \) for every \( i \), this means that \( f(\sum_i x_i) = g(\sum_i x_i) \).
Therefore \( \sum_i x_i \in |\smod'| \).
Since \( \smod' \) is a full submodule, the sum \( \sum_i x_i \) is also defined in \( \smod' \).

Suppose that \( h \colon \smodc \longrightarrow \smod \) is a \( \srig \)-linear map such that \( f \circ h = g \circ h \).
Let \( k \colon |\smodc| \longrightarrow |\smod'| \) be the unique function such that \( h = \iota \circ k \).
It suffices to show that \( k \) is \( \srig \)-linear, but it is easy to check this fact by using the above-discussed properties.

\subsection{Proof of Lemma~\ref{lem:module:completion}}
    Let \( \mathcal{G} = \{ \underline{x} \mid x \in |\smod| \} \) and \( \mathcal{R}_1 = \{ \sum_i \underline{x_i} = \underline{x} \mid \sum_i x_i = x \text{ in \( \smod \)} \} \cup \{ r \cdot \underline{x} = \underline{r \cdot x} \mid x \in |\smod| \} \).
    Let \( \mathcal{R}_2 \) be the relation consisting of \( \IsDefined{(\sum_i \underline{x_i})} \) for every countable family \( (x_i)_i \), for which the sum is not necessarily defined in \( \smod \).
    Then \( \CComp{\smod} = \langle \mathcal{G}, \mathcal{R}_1 \cup \mathcal{R}_2 \rangle \) with \( \smod \longrightarrow \CComp{\smod}, x \mapsto \underline{x} \) satisfies the requirement.

\subsection{Proof of Lemma~\ref{lem:module:completion-monic-epic}}
    Since \( |\smod| \) is a generator of \( \CComp{\smod} \), an \( \srig \)-linear map \( \CComp{\smod} \longrightarrow \smodb \) is completely determined by its action on \( |\smod| \); hence \( \smod \longrightarrow \CComp{\smod} \) is an epimorphism.
    In order to show that \( \smod \longrightarrow \CComp{\smod} \) is monic, it suffices to find a monomorphism \( \smod \hookrightarrow \smodb \) to a complete \( \srig \)-module \( \smodb \); then \( \smod \hookrightarrow \smodb \) factors through \( \smod \longrightarrow \CComp{\smod} \) by the universality of \( \CComp{\smod} \) and hence it is monic.
    We define a complete \( \Sigma \)-monoid \( \smod_{\infty} \): its underlying set is \( |\smod_\infty| \defe |\smod| \uplus \{\infty\} \); if \( \sum_i x_i = x \) in \( \smod \), the same equation holds in \( \smod_\infty \); and, if \( \sum_i x_i \) is not defined in \( \smod \), then \( \sum_i x_i = \infty \) in \( \smod_\infty \).
    Let \( \srig \multimap \smod_\infty \) be a set of linear maps, i.e.~homomorphisms between \( \Sigma \)-modules; then \( \srig \multimap \smod_\infty \) is a complete \( \Sigma \)-monoid.
    The action of \( \srig \) is defined by \( (r \cdot f)(r') \defe f(r r') \).
    Then \( (r \cdot (\sum_i f_i))(r') = (\sum_i f_i)(r r') = \sum_i f_i(r r') = \sum_i r \cdot f_i(r') \) and, if \( \sum_i r_i = r \), then \( ((\sum_i r_i) \cdot f)(r') = f((\sum_i r_i) r') = f(\sum_i r_i r') = \sum_i f(r_i r') = \sum_i (r_i \cdot f)(r') \).
    Given \( x \in |\smod| \), consider \( h(x)(r) \defe r \cdot x \).
    Then \( h \colon \smod \longrightarrow (\srig \multimap \smod_\infty) \).
    Actually \( h(r \cdot x)(r') = r' \cdot (r \cdot x) = (r r') \cdot x = (r \cdot h(x))(r') \) and, if \( \sum_i x_i = x \) in \( \smod \), then \( h(\sum_i x_i)(r') = r' \cdot (\sum_i x_i) = \sum_i r' \cdot x_i = \sum_i h(x_i)(r') = (\sum_i h(x_i))(r') \).
    This map \( h \) is injective since \( h(x) = h(y) \) implies \( x = h(x)(1) = h(y)(1) = y \).

\subsection{Proof of Lemma~\ref{lem:module:completion-dcfull}}
    We first show that \( |\smod| \subseteq |\CComp{\smod}| \) is full.
    Recall the \( \srig \)-module \( \srig \multimap \smod_{\infty} \) in the proof of Lemma~\ref{lem:module:completion-monic-epic} and the \( \srig \)-linear map \( h \colon \smod \longrightarrow (\srig \multimap \smod_{\infty}) \).
    Let \( (x_i)_i \) be a countable family of \( |\smod| \) and suppose that \( \sum_i x_i = y \) in \( \CComp{\smod} \) and \( y \in |\smod| \).
    Mapping this equation by the canonical map \( \CComp{\smod} \longrightarrow (\srig \multimap \smod_{\infty}) \), we have \( \sum_i h(x_i) = h(y) \).
    Hence \( \sum_i x_i = (\sum_i h(x_i))(1) = h(y)(1) = y \) holds in \( \smod_\infty \).
    Since \( y \neq \infty \), this means \( \sum_i x_i = y \) in \( \smod \).
    Consider \( x \in |\smod| \) and \( y \in |\CComp{\smod}| \) and suppose \( y \le x \) in \( \CComp{\smod} \), i.e.~\( y + z = x \) in \( \CComp{\smod} \).
    Since \( \CComp{\smod} \) is generated by \( |\smod| \), we have \( y = \sum_i r_i \cdot y_i \) and \( z = \sum_j r'_j \cdot z_j \) for some \( r_i,r'_j \in |\srig| \) and \( y_i,z_j \in |\smod| \).
    So \( \sum_i r_i \cdot y_i + \sum_j r'_j \cdot z_j = x \) in \( \CComp{\smod} \).
    Since \( r_i \cdot y_i, r'_j \cdot z_j \in |\smod| \), by the fullness of \( \smod \hookrightarrow \CComp{\smod} \), we have \( \sum_i r_i \cdot y_i + \sum_j r'_j \cdot z_j = x \) in \( \smod \) as well.

 \section{Supplementary materials for Section~\ref{sec:classic}}
In this section, we often identify an element \( e \in |\smod| \) of an \( \srig \)-module \( \smod \) with the \( \srig \)-linear map \( \srig \longrightarrow \smod, r \mapsto r \cdot e \).
For example, the condition
\begin{equation*}
    x
    \quad=\quad
    \sum_i \varphi_i(x) \cdot e_i
\end{equation*}
for a countable dual basis \( (e_i, \varphi_i)_{i \in I} \) is written as
\begin{equation*}
    x
    \quad=\quad
    \sum_i (e_i \circ \varphi_i)(x)
\end{equation*}
or equivalently
\begin{equation*}
    \ident_\smod
    \quad=\quad
    \sum_i e_i \circ \varphi_i.
\end{equation*}

\subsection{Technical lemmas}
The first lemma is an explicit description of the equaliser of a set of parallel morphisms.
\begin{lemma}
    For every collection \( F \subseteq \SMod[\srig](\smod,\smodb) \) of parallel morphisms, its equalizer is the sum-reflecting submodule \( \iota \colon X \hookrightarrow \smod \) of \( \smod \) consisting of \( \{ x \in |\smod| \mid \forall f,g \in F. f(x) = g(x) \} \).
\end{lemma}
\begin{proof}
    Let us first prove that \( X \) is actually a \( \srig \)-module, i.e.~it is closed under the action of \( \srig \) in \( \smod \).
    If \( x \in |X| \), then \( f(x) = g(x) \) for every \( f, g \in F \).
    So \( f(r \cdot x) = r \cdot f(x) = r \cdot g(x) = g(r \cdot x) \) for every \( f,g \in F \), which implies \( r \cdot x \in |X| \).

    It is easy to see that \( f \circ \iota = g \circ \iota \) for every \( f,g \in F \).

    Let \( h \colon Y \longrightarrow \smod \) and assume that \( f \circ h = g \circ h \) for every \( f,g \in F \).
    Then \( h(y) \in |X| \) for every \( y \in |Y| \) and thus we have a unique function \( k \colon |Y| \longrightarrow |X| \) such that \( \iota \circ k = h \) as a function on underlying sets.
    It suffices to show that \( k \) is \( \srig \)-linear.
    Let \( y \in |Y| \).
    We have
    \begin{equation*}
        \iota(k(r \cdot y)) = h(r \cdot y) = r \cdot h(y) = r \cdot \iota(k(y)) = \iota(r \cdot k(y)).
    \end{equation*}
    Since \( \iota \) is injective, \( k(r \cdot y) = r \cdot k(y) \).
    Suppose that \( \sum_i y_i \) converges in \( Y \).
    Then
    \begin{equation*}
        \iota(k(\sum_i y_i))
        \quad\Keq\quad
        \sum_i \iota(k(y_i)).
    \end{equation*}
    Since \( \iota \colon X \hookrightarrow \smod \) is a sum-reflecting submodule by definition, \( \sum_i k(y_i) \) converges in \( X \) and
    \begin{equation*}
        \iota(\sum_i k(y_i))
        \quad=\quad
        \sum_i \iota(k(y_i)).
    \end{equation*}
    Hence \( \iota(k(\sum_i y_i)) = \iota(\sum_i k(y_i)) \).
    Since \( \iota \) is injective, \( k(\sum_i y_i) = \sum_i k(y_i) \).
\end{proof}

The next three lemmas are about the tensor product \( ({-}) \otimes \smod \) with a module \( \smod \) having a countable dual basis.
A module \( \smod \) \emph{preserves monomorphisms} if the tensor product \( ({-}) \otimes \smod \) preserves monomorphisms, i.e.~\( f \otimes \smod \colon \smodb \otimes \smod \longrightarrow \smodc \otimes \smod \) is a monomorphism for every monomorphism \( f \colon \smodb \hookrightarrow \smodc \).
\begin{lemma}\label{lem:appx:submodule}
    Every module \( \smod \) with a countable dual basis preserves monomorphisms.
\end{lemma}
\begin{proof}
    Suppose that \( \smod \) has a basis \( (e_i, \varphi_i)_{i \in I} \) and let \( f \colon \smodb \hookrightarrow \smodb' \) be a monomorphism.
    Since \( \smodb \otimes ({-}) \) preserves sums, we have
    \begin{equation*}
        \smodb \otimes \ident_{\smod} = \smodb \otimes (\sum_{i \in I} e_i \circ \varphi_i) \Kle \sum_{i \in I} (\smodb \otimes e_i) \circ (\smodb \otimes \varphi_i).
    \end{equation*}
    This implies that the family \( (\smodb \otimes \varphi_i \colon \smodb \otimes \smod \longrightarrow \smodb \otimes \srig)_{i \in I} \) of morphisms is jointly monic.

    Suppose that \( (f \otimes \smod) \circ g = (f \otimes \smod) \circ h \) for some \( g,h \colon X \longrightarrow \smod \).
    Then we have
    \begin{align*}
       (f \otimes \srig) \circ (\smodb \otimes \varphi_i) \circ g
       &=
       (\smodb' \otimes \varphi_i) \circ  (f \otimes \smod) \circ g
       \\
       &=
       (\smodb' \otimes \varphi_i) \circ  (f \otimes \smod) \circ h
       \\
       &=
       (f \otimes \srig) \circ (\smodb \otimes \varphi_i) \circ h
    \end{align*}
    for every \( i \in I \).
    Since \( f \otimes \srig \) as the composite \( \smodb \otimes \srig \cong \smodb \stackrel{f}{\longrightarrow} \smodb' \cong \smodb' \otimes \srig \) of monomorphisms is monic, the above equation leads to
    \begin{equation*}
        (\smodb \otimes \varphi_i) \circ g
        \;\;=\;\;
        (\smodb \otimes \varphi_i) \circ h
    \end{equation*}
    for every \( i \in I \).
    Since \( (\smodb \otimes \varphi_i)_{i \in I} \) is jointly monic, we conclude \( g = h \).
\end{proof}

\begin{lemma}\label{lem:appx:tensor-sr}
    Let \( \smod \) be a module with a basis.
    If \( f \colon \smodb \hookrightarrow \smodb' \) is a sum-reflecting submodule, then so is \( f \otimes \smod \).
\end{lemma}
\begin{proof}
    Suppose that \( \smod \) has a basis \( (e_i, \varphi_i)_{i \in I} \) and let \( f \colon \smodb \hookrightarrow \smodb' \) be a sum-reflecting submodule.
Suppose that \( \sum_j (f \otimes \smod)(x_j) \) is defined in the module \( \smodb' \otimes \smod \) and \( (f \otimes \smod)(y) = \sum_j (f \otimes \smod)(x_j) \) for some \( y \in |\smodb \otimes \smod| \).
    Then, for every \( i \in I \),
    \begin{equation*}
        ((\smodb' \otimes \varphi_i) \circ (f \otimes \smod))(y) = \sum_j ((\smodb' \otimes \varphi_i) \circ (f \otimes \smod))(x_j).
    \end{equation*}
    Hence
    \begin{equation*}
        ((f \otimes \srig) \circ (\smodb \otimes \varphi_i))(y) = \sum_j ((f \otimes \srig) \circ (\smodb \otimes \varphi_i))(x_j)
    \end{equation*}
    for every \( i \in I \).
    Since \( f \) is a sum-reflecting submodule, so is \( f \otimes \srig \).
    Hence, for every \( i \in I \),
    \begin{equation*}
        (\smodb \otimes \varphi_i)(y) = \sum_j (\smodb \otimes \varphi_i)(x_j)
    \end{equation*}
    So
    \begin{equation*}
        y = \sum_i (\smodb \otimes e_i)((\smodb \otimes \varphi_i)(y)) = \sum_i \sum_j (\smodb \otimes e_i)((\smodb \otimes \varphi_i)(x_j))
    \end{equation*}
    and
    \begin{equation*}
        \sum_j x_j = \sum_j \sum_i (\smodb \otimes e_i)((\smodb \otimes \varphi_i)(x_j)).
    \end{equation*}
    Hence \( y = \sum_j x_j \).
\end{proof}

\begin{lemma}\label{lem:appx:tensor-dcsr}
    Let \( \smod \) be a module with a basis.
    If \( f \colon \smodb \hookrightarrow \smodb' \) is a downward-closed sum-reflecting submodule, then so is \( f \otimes \smod \).
\end{lemma}
\begin{proof}
    Suppose that \( \smod \) has a basis \( (e_i, \varphi_i)_{i \in I} \) and let \( f \colon \smodb \hookrightarrow \smodb' \) be a downward-closed sum-reflecting submodule.
    Then \( f \otimes \smod \) is sum-reflecting by Lemma~\ref{lem:appx:tensor-sr}.
    We prove that it is downward-closed.

    We first show that, for every \( x,y \in \smodb \otimes \smod \),
    \begin{equation*}
        x \le y
        \quad\Leftrightarrow\quad
        \forall i \in I.\quad (\smodb \otimes \varphi_i)(x) \le (\smodb \otimes \varphi_i)(y).
    \end{equation*}
    The \((\Rightarrow)\) direction follows from the monotonicity of \( \srig \)-linear maps.
    The other direction comes from the fact that
    \begin{equation*}
        \smodb \otimes \ident_{\smod} = \sum_{i \in I} (\smodb \otimes e_i) \circ (\smodb \otimes \varphi_i).
    \end{equation*}
    Actually, if \( (\smodb \otimes \varphi_i)(x) \le (\smodb \otimes \varphi_i)(y) \) for every \( i \in I \), then \( (\smodb \otimes e_i) \circ (\smodb \otimes \varphi_i)(x) \le (\smodb \otimes e_i) \circ (\smodb \otimes \varphi_i)(y) \) for every \( i \in I \).
    Hence
    \begin{equation*}
        x = \sum_i (\smodb \otimes e_i) \circ (\smodb \otimes \varphi_i)(x)
        \le \sum_i (\smodb \otimes e_i) \circ (\smodb \otimes \varphi_i)(y) = y.
    \end{equation*}
    By the same argument, for every \( x',y' \in \smodb' \otimes \smod \),
    \begin{equation*}
        x' \le y'
        \quad\Leftrightarrow\quad
        \forall i \in I.\quad (\smodb' \otimes \varphi_i)(x) \le (\smodb' \otimes \varphi_i)(y).
    \end{equation*}

    Suppose that \( (f \otimes \smod)(y) = y' \) and \( x' \le y' \) in \( \smodb' \otimes \smod \).
    Then, for every \( i \in I \), \( (\smodb' \otimes \varphi_i)((f \otimes \smod)(y)) = (\smodb' \otimes \varphi_i)(y') \ge (\smodb' \otimes \varphi_i)(x') \).
    Since \( (\smodb' \otimes \varphi_i) \circ (f \otimes \smod) = (f \otimes \srig) \circ (\smod \otimes \varphi_i) \), we have
    \begin{equation*}
        (f \otimes \srig)(y_i) \ge (\smodb' \otimes \varphi_i)(x'),
    \end{equation*}
    where \( y_i = (\smod \otimes \varphi_i)(y) \in (\smod \otimes \srig) \).
    Since \( f \) is a downward-closed sum-reflecting submodule, so is \( f \otimes \srig \).
    By the downward-closedness, there exists \( x_i \in \smodb \otimes \srig \) such that
    \begin{equation*}
        (f \otimes \srig)(x_i) = (\smodb' \otimes \varphi_i)(x')
    \end{equation*}    
    for each \( i \in I \); by the sum-reflection, \( x_i \le y_i \) holds in \( \smodb \otimes \srig \).
    Then
    \begin{equation*}
        x'
        = \sum_i (\smodb' \otimes e_i) \circ (\smodb' \otimes \varphi_i)(x')
        = \sum_i (\smodb' \otimes e_i)((f \otimes \srig)(x_i)).
    \end{equation*}
    Hence
    \begin{equation*}
        x'
        = \sum_i (f \otimes \smod) ((\smodb \otimes e_i)(x_i)).
    \end{equation*}
    Since \( y = \sum_i (\smodb \otimes e_i)(\smodb \otimes \varphi_i)(y) = \sum_i (\smodb \otimes e_i)(y_i) \) converges and \( x_i \le y_i \) for every \( i \), we know that \( \sum_i (\smodb \otimes e_i)(x_i) \) converges.
    Let \( x = \sum_i (\smodb \otimes e_i)(x_i) \).
    Then \( x' = (f \otimes \smod)(x) \) as desired.
    Hence \( f \otimes \smod \) is downward-closed.
\end{proof}

\subsection{Proof of Lemma~\ref{lem:classic:vector-idempotent}}
It suffices to show that the multiplication \( \mu_V \colon V^{\bot\bot\bot\bot} \longrightarrow V^{\bot\bot} \) of the monad is monic for every \( V \in \SVec[\srig] \).
As in any other SMCC, \( \mu_V(k) = \lambda f. k\,(\lambda h. h(f)) \).
Let \( \eta_{V^{\bot\bot}} \colon V^{\bot\bot} \longrightarrow V^{\bot\bot\bot\bot} \) by \( \eta_{V^{\bot\bot}}(g)(h) \defe h(g) \).
We prove \( \eta_{V^{\bot\bot}} \circ \mu_V = \ident_{V^{\bot\bot\bot\bot}} \).

Let \( (e_i, \varphi_i)_{i \in I} \) be a basis of \( V \).
Then \( (\varphi_i, \hat{e}_i)_{i \in I} \) is a basis of \( V^{\bot} \), where \( \hat{e}_i(f) = f(e_i) \).
Similarly, \( (\hat{e}_i, \hat{\varphi}_i)_{i \in I} \), \( (\hat{\varphi}_i, \hat{\hat{e}}_i)_{i \in I} \) and \( (\hat{\hat{e}}_i, \hat{\hat{\varphi}}_i)_{i \in I} \) are bases of \( V^{\bot\bot} \), \( V^{\bot\bot\bot} \) and \( V^{\bot\bot\bot\bot} \), respectively.

The goal is to prove that \( \hat{\hat{\varphi}}_i(\eta_{V^{\bot\bot}}(\mu_V(\hat{\hat{e}}_j)) = \hat{\hat{\varphi}}_i(\hat{\hat{e}}_j) = \delta_{i,j} \) (where \( \delta_{i,i} = 1 \) and \( \delta_{i,j} = 0 \) if \( i \neq j \)) since \( (\hat{\hat{e}}_i, \hat{\hat{\varphi}}_i)_{i \in I} \) is a basis of \( V^{\bot\bot\bot\bot} \).
Expressing \( \eta_{V^{\bot\bot}}(g) \) by using the basis,
\begin{align*}
    &
    \textstyle
    \eta_{V^{\bot\bot}}(g)
    = \sum_{k} \hat{\hat{\varphi}}_k(\eta_{V^{\bot\bot}}(g)) \cdot \hat{\hat{e}}_k
    \\
    &\quad\textstyle
    = \sum_k \eta_{V^{\bot\bot}}(g)(\hat{\varphi}_k) \cdot \hat{\hat{e}}_k
    = \sum_k \hat{\varphi}_k(g) \cdot \hat{\hat{e}}_k.
\end{align*}
Hence
\begin{align*}
    \textstyle
    \hat{\hat{\varphi}}_i(\eta_{V^{\bot\bot}}(\mu_V(\hat{\hat{e}}_j))
    &= \textstyle
    \hat{\hat{\varphi}}_i(\sum_k \hat{\varphi}_k(\mu_V(\hat{\hat{e}}_j)) \cdot \hat{\hat{e}}_k)
    \\
    &\Kle \textstyle
    \sum_k \hat{\varphi}_k(\mu_V(\hat{\hat{e}}_j)) \cdot \hat{\hat{\varphi}}_i(\hat{\hat{e}}_k)
    \\
    &= \textstyle
    \hat{\varphi_i}(\mu_V(\hat{\hat{e}}_j))
\end{align*}
and
\begin{equation*}
    \textstyle
    \hat{\varphi_i}(\mu_V(\hat{\hat{e}}_j))
    =
    \mu_V(\hat{\hat{e}}_j)(\varphi_i)
    =
    \hat{\hat{e}}_j(\lambda h. h(\varphi_i))
    =
    \hat{e}_j(\varphi_i)
    = \delta_{i,j}
\end{equation*}
as required.

\subsection{Proof of Theorem~\ref{thm:classic:vector}}
    If \( V,W \in \SVec[\srig]^{\bot\bot} \), then \( (W \multimap V) \in \SVec[\srig]^{\bot\bot} \) since \( (W \multimap V) \cong (W \multimap V^{\bot\bot}) \cong (W \otimes V^{\bot})^\bot \) and \( (W \otimes V^{\bot}) \in \SVec[\srig] \).
    For every \( V,W \in \SVec[\srig] \) and \( U \in \SVec[\srig]^{\bot\bot} \), we have
    \begin{align*}
        \Bilin(V,W; U)
        &\cong \SVec[\srig](V \otimes W, U) \\
        &\cong \SVec[\srig]((V \otimes W)^{\bot\bot}, U)
    \end{align*}
    natural in \( U \in \SVec[\srig]^{\bot\bot} \).
    Hence \( (V \otimes W)^{\bot\bot} \) is a representing object of bilinear maps in \( \SVec[\srig]^{\bot\bot} \).
    Since
    \begin{equation*}
        \SVec[\srig](V, W \multimap U) \cong \Bilin(V,W; U)
    \end{equation*}
    natural in \( U \), we have \( (({-})\otimes W)^{\bot\bot} \dashv W \multimap ({-}) \).
    So we have an SMCC with \( U^{\bot\bot} \cong U \), i.e.~a model of classical MLL.
    As a left adjoint, \( ({-})^{\bot\bot} \) preserves coproducts, so \( \SVec[\srig]^{\bot\bot} \) has additives.

    Suppose that \( !V \in \SVec[\srig] \) for every \( V \in \SVec[\srig] \).
    Since \( \SVec[\srig] \) is a full subcategory of \( \SMod[\srig] \), this means that \( !V \) is (the carrier object of) the cofree cocommutative comonid in \( \SVec[\srig] \).
    Hence
    \begin{equation*}
        \xymatrix{
            \Comon(\SVec[\srig])
            \ar@{}[r]|(0.6){\bot}
            \ar@<5pt>[r]
            & \SVec[\srig] \ar@<5pt>[l] \ar@<5pt>[r] \ar@{}[r]|(0.47){\bot} & \SVec[\srig]^{\bot\bot} \ar@<5pt>[l]
        }
    \end{equation*}
    and left adjoints are strong monoidal.
    Hence \( \SVec[\srig]^{\bot\bot} \) is a linear-non-linear model of classical linear logic.

\subsection{Proofs of Theorems~\ref{thm:classical:vector:exponential} and \ref{thm:classical:proj:exponential}}
\label{sec:appx:classic:exponential-proof}

Let \( \srig \) be a \( \Sigma \)-semiring and assume that every downward-closed sum-reflecting ideal of \( \srig \) has a countable dual basis.
Sometimes we further assume that every downward-closed sum-reflecting ideal of \( \srig \) has a countable orthogonal dual basis; this situation is called the \emph{orthogonality assumption}.

\begin{lemma}\label{lem:appx:basis-of-symmetric-tensor-power}
    For every module \( V \) with a countable dual basis, the symmetric tensor power \( \SymTensor{V}{n} \) also has a countable  dual basis.
    Furthermore, under the orthogonality assumption, if \( V \) has a countable orthogonal dual basis, then so does \( \SymTensor{V}{n} \).
\end{lemma}
\begin{proof}
    Suppose that \( V \) has a countable dual basis \( (e_i, \varphi_i)_{i \in I} \).
    For a sequence \( s = (s_1,\dots,s_n) \in I^n \), we define \( e_s \defe (e_{s_1} \otimes \dots \otimes e_{s_n}) \) and \( \varphi_s(x_1 \otimes \dots \otimes x_n) \defe \varphi_{s_1}(x_1) \dots \varphi_{s_n}(x_n) \).
    Then \( (e_s, \varphi_s)_{s \in I^n} \) is a basis of \( V^{\otimes n} \).
    Hence
    \begin{equation*}
        \ident_{V^{\otimes n}} \quad=\quad \sum_{s \in I^n} e_s \circ \varphi_s,
    \end{equation*}
    where \( e_s \in |V^{\otimes n}| \) is seen as \( e_s \colon \srig \longrightarrow V^{\otimes n} \) by \( e_s(r) \defe r \cdot e_s \).

    For a permutation \( \sigma \in \mathfrak{S}_n \) and \( s \in I^n \), we define \( \sigma(s) = (s_{\sigma^{-1}(1)}, \dots, s_{\sigma^{-1}(n)}) \).
    Then \( \sigma(e_{s}) = e_{\sigma(s)} \) and \( \varphi_s \circ \sigma = \varphi_{\sigma^{-1}(s)} \).

    Let \( \mathcal{M}_n(I) \) be the set of finite multisets over \( I \) of size \( n \).
    We write \( s \lhd \xi \) if \( s \) is a sequence obtained by arranging elements of \( m \).
    For every \( x \in \SymTensor{V}{n} \) and permutation \( \sigma \in \mathfrak{S}_n \), \( \sigma(x) = x \) and thus \( \varphi_s(x) = \varphi_s(\sigma(x)) = \varphi_{\sigma^{-1}(s)}(x) \).
    The above observation says that \( \varphi_\xi(x) \) is well-defined for every \( \xi \in \mathcal{M}_n(I) \) and \( x \in |\SymTensor{V}{n}| \).
    For every \( x \in |\SymTensor{V}{n}| \), we have
    \begin{align*}
        x
        &=
        \sum_{s \in I^{n}} \varphi_s(x) \cdot e_s
        \\
        &=
        \sum_{\xi \in \mathcal{M}_n(I)} \sum_{s \lhd \xi} \varphi_s(x) \cdot e_s
        \\
        &=
        \sum_{\xi \in \mathcal{M}_n(I)} \sum_{s \lhd \xi} \varphi_\xi(x) \cdot e_s
    \end{align*}
    in \( V^{\otimes n} \) and, in particular, \( (\sum_{s \lhd \xi} \varphi_\xi(x) \cdot e_s) \in |\SymTensor{V}{n}| \) where the sum is taken in \( V^{\otimes n} \).

    For \( \xi \in \mathcal{M}_n(I) \), let \( R_\xi \) be a subset of \( |\srig| \) defined by
    \begin{equation*}
        R_\xi
        \quad\defe\quad
        \{ r \in |\srig| \mid \IsDef{(\sum_{s \lhd \xi} r \cdot e_s)} \mbox{ in \( V^{\otimes n} \)} \}.
    \end{equation*}
    \( R_\xi \) is obviously downward-closed.
    Hence \( R_\xi \) can be seen as a downward-closed sum-reflecting ideal of \( \srig \).
    By the above argument, \( \varphi_\xi(x) \in R_\xi \) for every \( x \in |\SymTensor{V}{n}| \).
    By definition, the mapping
    \begin{equation*}
        e_\xi \colon R_\xi \longrightarrow \SymTensor{V}{n},
        \qquad
        r \mapsto \sum_{s \lhd \xi} r \cdot e_s
    \end{equation*}
    is well-defined.  It is easy to see that \( e_\xi \) is \( \srig \)-linear.

    By the assumption, \( R_\xi \) has a countable dual basis \( (e_{\xi,j}, \varphi_{\xi,j})_{j \in J_\xi} \) and \( r = \sum_{j \in J_\xi} \varphi_{\xi,j}(r) \cdot e_{\xi,j} \) holds in \( R_\xi \) for every \( r \in R_\xi \).
    Then, for every \( x \in |\SymTensor{V}{n}| \),
    \begin{align*}
        x
        &=
        \sum_{\xi \in \mathcal{M}_n(I)} \sum_{s \lhd \xi} \varphi_\xi(x) \cdot e_s
        \\
        &=
        \sum_{\xi \in \mathcal{M}_n(I)} (e_\xi \circ \varphi_\xi)(x)
        \\
        &=
        \sum_{\xi \in \mathcal{M}_n(I)} (e_\xi \circ (\sum_{j \in J_\xi} e_{\xi,j} \circ \varphi_{\xi,j}) \circ \varphi_\xi)(x)
        \\
        &\Kle
        \sum_{\xi \in \mathcal{M}_n(I)} \sum_{j \in J_\xi} (e_\xi \circ e_{\xi,j} \circ \varphi_{\xi,j} \circ \varphi_\xi)(x)
    \end{align*}
    where the sum is taken in \( V^{\otimes n} \).
    The equation
    \begin{equation*}
        x = \sum_{\xi \in \mathcal{M}_n(I)} \sum_{j \in J_\xi} (e_\xi \circ e_{\xi,j} \circ \varphi_{\xi,j} \circ \varphi_\xi)(x)
    \end{equation*}
    contains only values in \( \SymTensor{V}{n} \) and \( \SymTensor{V}{n} \) is sum-reflecting, this equation holds in \( \SymTensor{V}{n} \) as well.
    This means that \( \SymTensor{V}{n} \) has a countable dual basis \( (e_\xi \circ e_{\xi,j}, \varphi_{\xi,j} \circ \varphi_\xi)_{\xi \in \mathcal{M}_n, j \in J_\xi} \), where \( e_\xi \circ e_{\xi,j} \colon \srig \longrightarrow \SymTensor{V}{n} \) is identified with an element \( (e_\xi \circ e_{\xi,j})(1) \in |\SymTensor{V}{n}| \).

    Let us consider the orthogonal case.
    Then the bases \( (e_s, \varphi_s)_{s \in I^n} \) and \( (e_{\xi,j}, \varphi_{\xi,j})_{j \in J_\xi} \) can be assume to be orthogonal.
    Then
    \begin{equation*}
        \varphi_{\xi,j} \circ \varphi_\xi \circ e_\xi \circ e_{\xi,j}
        \quad=\quad
        \varphi_{\xi,j} \circ e_{\xi,j}
        \quad=\quad
        \ident.
    \end{equation*}
    If \( \xi \neq \xi' \), then
    \begin{equation*}
        \varphi_{\xi,j} \circ \varphi_\xi \circ e_{\xi'} \circ e_{\xi',j'}
        \quad=\quad
        \varphi_{\xi,j} \circ 0 \circ e_{\xi',j'}
        \quad=\quad
        0.
    \end{equation*}
    If \( j \neq j' \), then
    \begin{equation*}
        \varphi_{\xi,j} \circ \varphi_\xi \circ e_{\xi} \circ e_{\xi,j'}
        \quad=\quad
        \varphi_{\xi,j} \circ e_{\xi,j'}
        \quad=\quad
        0.
    \end{equation*}
    Hence \( (e_\xi \circ e_{\xi,j}, \varphi_{\xi,j} \circ \varphi_\xi)_{\xi \in \mathcal{M}_n, j \in J_\xi} \) is orthogonal.
\end{proof}

Let \( !V \) be the cofree cocommutative comonoid and
\[
    d^{(n)} \colon !V \longrightarrow (!V)^{\otimes n} \longrightarrow V^{\otimes n}
\]
for each \( n \in \mathbb{N} \).
By the cocommutativity, \( d^{(n)} \) factors as \( !V \stackrel{\hat{d}^{(n)}}{\longrightarrow}\SymTensor{V}{n} \hookrightarrow V^{\otimes n} \).
We define \( \hat{d} \colon {!}V \longrightarrow \prod_{n < \omega} \SymTensor{V}{n} \) by \( \hat{d}(x) = \sum_{n < \omega} \hat{d}^{(n)}(x) \).
Let \( \S V \) and \( \S^{(n)}V \) be sum-reflecting submodules of \( \prod_{n < \omega} \SymTensor{V}{n} \) and \( \SymTensor{V}{n} \) consisting of
\begin{align*}
    |\S V| &:= \{ x \in |\prod_{n < \omega} \SymTensor{V}{n}| \mid \exists y \in |{!}V|. x \le \hat{d}(y) \} \\
    |\S^{(n)} V| &:= \{ x \in |\SymTensor{V}{n}| \mid \exists y \in |{!}V|. x \le \hat{d}^{(n)}(y) \}.
\end{align*}
\begin{lemma}\label{lem:appx:base-of-section}
    \( \S V \) and \( \S^{(n)} V \) have countable dual bases.
    Furthermore, under the orthogonal assumption, if \( V \) has a countable orthogonal dual basis, then so do  \( \S V \) and \( \S^{(n)} V \).
\end{lemma}
\begin{proof}
    By the previous lemma, \( \SymTensor{V}{n} \) has a countable dual basis \( (e_{n,i}, \varphi_{n,i})_{i \in I_n} \).
    Let \( R_{n,i} \) be the downward-closure of the image of \( \varphi_{n,i} \circ d^{(n)} \colon {!}V \longrightarrow \srig \), i.e.~\( R_{n,i} \defe \{ r \in |\srig| \mid \exists x \in !V. r \le \varphi_{n,i}(\hat{d}^{(n)}(x)) \} \).
    Then \( R_{n,i} \) is closed under the multiplication of \( \srig \) (i.e.~\( R_{n,i} \) is an ideal) and
    \begin{equation*}
        e_{n,i} \colon R_{n,i} \longrightarrow \S^{(n)} V,
        \qquad
        r \mapsto r \cdot e_{n,i}
    \end{equation*}
    is an \( \srig \)-linear map.
    Since \( \S^{(n)} \) is a sum-reflecting submodule of \( \SymTensor{V}{n} \) and
    \begin{equation*}
        x
        \quad=\quad
        \sum_{i \in I_n} \varphi_{n,i}(x) \cdot e_{n,i}
    \end{equation*}
    holds in \( \SymTensor{V}{n} \), we have
    \begin{equation*}
        \ident_{\S^{(n)} V}
        \quad=\quad
        \sum_{i \in I_n} e_{n,i} \circ \varphi_{n,i}
    \end{equation*}
    holds in \( \SMod[\srig](\S^{(n)}V, \S^{(n)}V) \).
    By the assumption, \( R_{n,i} \) has a countable dual basis, say \( (e_{n,i,j}, \varphi_{n,i,j})_{j \in J_{n,i}} \).
    So
    \begin{equation*}
        \ident_{\S^{(n)} V}
        \quad=\quad
        \sum_{i \in I_n, j \in J_{n,i}} e_{n,i} \circ e_{n,i,j} \circ \varphi_{n,i,j} \circ \varphi_{n,i},
    \end{equation*}
    which means that \( (e_{n,i,j}, \varphi_{n,i,j})_{i \in I_n, j \in J_{n,i}} \) is a countable dual basis of \( \S^{(n)} V \).

    Under the orthogonality assumption, if \( V \) has a countable orthogonal dual basis, then \( \SymTensor{V}{n} \) has a countable orthogonal dual basis by the previous lemma.
    Then \( (e_{n,i}, \varphi_{n,i})_{i \in I_n} \) and \( (e_{n,i,j}, \varphi_{n,i,j})_{j \in J_{n,i}} \) can be assumed to be orthogonal.
    In this case, \( (e_{n,i,j}, \varphi_{n,i,j})_{i \in I_n, j \in J_{n,i}} \) is orthogonal.

    The claim for \( \S V \) follows from the fact that the collection of
    \begin{align*}
        \pi_n &\quad\colon\quad (\prod_{n < \omega} \SymTensor{V}{n}) \longrightarrow \SymTensor{V}{n} \\
        \iota_n &\quad\colon\quad \SymTensor{V}{n} \longrightarrow (\prod_{n < \omega} \SymTensor{V}{n})
    \end{align*}
    behaves like a basis, that means, \( \ident = \sum_{n < \omega} \iota_n \circ \pi_n \).
    These morphisms can be restricted to
    \begin{align*}
        \pi_n &\quad\colon\quad \S V \longrightarrow \S^{(n)} V \\
        \iota_n &\quad\colon\quad \S^{(n)} V \longrightarrow \S V,
    \end{align*}
    and the equation \( \ident = \sum_{n < \omega} \iota_n \circ \pi_n \) still holds for the restrictions.
    This equation together with the bases of \( \S^{(n)} \) constructed above induces a basis of \( \S V \).
\end{proof}

Let \( \SymTensor{V}{n,m} \) be the equaliser of \( \{ \sigma \otimes \sigma' \colon V^{\otimes n} \otimes V^{\otimes m} \longrightarrow V^{\otimes n} \otimes V^{\otimes m} \mid \sigma \in \mathfrak{S}_n, \sigma' \in \mathfrak{S}_m \} \) in \( \SMod[\srig] \).
Then there exists the canonical morphism \( \SymTensor{V}{n} \otimes \SymTensor{V}{m} \longrightarrow \SymTensor{V}{n,m} \).
\begin{lemma}\label{lem:appx:key-on-comultiplication}
    Let \( V \) be an \( \srig \)-module with a countable dual basis.
    The canonical morphism \( \SymTensor{V}{n} \otimes \SymTensor{V}{m} \longrightarrow \SymTensor{V}{n,m} \) is a monomorphism.
    Furthermore it is a downward-closed sum-reflecting submodule.
\end{lemma}
\begin{proof}
    Let us first describe the canonical map.
    Since \( \SymTensor{V}{n} \hookrightarrow V^{\otimes n} \) and \( \SymTensor{V}{m} \longrightarrow V^{\otimes m} \) are equalisers, they are sum-reflecting submodules.
    By the assumption and Lemmas~\ref{lem:appx:basis-of-symmetric-tensor-power} and \ref{lem:appx:tensor-sr}, their tensor product \( \SymTensor{V}{n} \otimes \SymTensor{V}{m} \longrightarrow V^{\otimes n} \otimes V^{\otimes m} \) is also a sum-reflecting submodule.
    Since \( (\sigma \otimes \sigma')(x \otimes x') = x \otimes x' \) for every \( x \in |\SymTensor{V}{n}| \) and \( x' \in |\SymTensor{V}{m}| \), this map canonically factors through \( \SymTensor{V}{n,m} \hookrightarrow V^{\otimes n} \otimes V^{\otimes m} \) as \( \SymTensor{V}{n} \otimes \SymTensor{V}{m} \longrightarrow \SymTensor{V}{n,m} \hookrightarrow V^{\otimes n} \otimes V^{\otimes m} \).
    Since the composite is a sum-reflecting submodule, so is \( \SymTensor{V}{n} \otimes \SymTensor{V}{m} \longrightarrow \SymTensor{V}{n,m} \).
    Hence it suffices to show that \( \SymTensor{V}{n} \otimes \SymTensor{V}{m} \longrightarrow \SymTensor{V}{n,m} \) is downward-closed.

    Let \( (e_i, \varphi_i)_{i \in I} \) be a countable dual basis of \( V \).
    We use the notions introduced in the proof of Lemma~\ref{lem:appx:basis-of-symmetric-tensor-power}.
    For \( \xi \in \mathcal{M}_n(I) \) and \( \zeta \in \mathcal{M}_m(I) \),
    \begin{align*}
        R_\xi &\defe \{ r \in |\srig| \mid \IsDef{(\sum_{s \lhd \xi} r \cdot e_s)} \mbox{ in \( V^{\otimes n} \)} \} \\
        R_\zeta &\defe \{ r \in |\srig| \mid \IsDef{(\sum_{t \lhd \zeta} r \cdot e_t)} \mbox{ in \( V^{\otimes m} \)} \} \\
        R_{\xi,\zeta} &\defe \{ r \in |\srig| \mid \IsDef{(\sum_{s \lhd \xi, t \lhd \zeta} r \cdot (e_s \otimes e_t))} \mbox{ in \( V^{\otimes n} \otimes V^{\otimes m} \)} \}.
    \end{align*}
    Then \( R_\xi \), \( R_\zeta \) and \( R_{\xi,\zeta} \) are downward-closed sum-reflecting ideals of \( \srig \).
    By the assumption, \( R_\xi \) and \( R_\zeta \) have countable dual bases.
    Hence \( R_\xi \otimes R_\zeta \longrightarrow \srig \otimes \srig \cong \srig \) is again a downward-closed sum-reflecting ideal by Lemma~\ref{lem:appx:tensor-dcsr}.
    Since \( r \in R_\xi \) and \( r' \in R_\zeta \) implies \( r r' \in R_{\xi,\zeta} \), by Lemma~\ref{lem:module:tensor-peak},
    we have \( R_\xi \otimes R_\zeta \hookrightarrow R_{\xi,\zeta} \) as downward-closed sum-reflecting ideals.

    We have mappings
    \begin{align*}
        e_\xi &\colon R_\xi \longrightarrow \SymTensor{V}{n}, & & r \mapsto \sum_{s \lhd \xi} r \cdot e_s \\
        e_\zeta &\colon R_\zeta \longrightarrow \SymTensor{V}{m}, & & r \mapsto \sum_{t \lhd \zeta} r \cdot e_t \\
        e_{\xi,\zeta} &\colon R_{\xi,\zeta} \longrightarrow \SymTensor{V}{n,m}, & & r \mapsto \sum_{s \lhd \xi, t \lhd \zeta} r \cdot (e_s \otimes e_t).
    \end{align*}
    Then \( R_{\xi,\zeta} \) contains the image of \( \varphi_\xi \otimes \varphi_\zeta \colon \SymTensor{V}{n,m} \longrightarrow \srig \) and
    \begin{align*}
\ident_{\SymTensor{V}{n,m,}} &= \sum_{\xi \in \mathcal{M}_n(I), \zeta \in \mathcal{M}_m(I)} e_{\xi,\zeta} \circ (\varphi_\xi \otimes \varphi_\zeta).
    \end{align*}
    We also have maps \( e_\xi \otimes e_\zeta \colon R_\xi \otimes R_\zeta \longrightarrow \SymTensor{V}{n} \otimes \SymTensor{V}{m} \), which satisfies
    \begin{align*}
        (e_\xi \otimes e_\zeta)(r \otimes r')
        &=
        e_{\xi}(r) \otimes e_{\zeta}(r')
        \\
        &=
        (\sum_{s \lhd \xi} r \cdot e_s) \otimes (\sum_{t \lhd \zeta} r' \cdot e_t)
        \\
        &\Kle
        \sum_{s \lhd \xi} \sum_{t \lhd \zeta} (r r') \cdot (e_s \otimes e_t)
        \\
        &=
        e_{\xi,\zeta}(r r')
    \end{align*}
    in \( V^{\otimes n} \otimes V^{\otimes m} \).
    Since \( \SymTensor{V}{n,m} \hookrightarrow V^{\otimes n} \otimes V^{\otimes m}\) is sum-reflecting, the equation
    \begin{equation*}
        (e_\xi \otimes e_\zeta)(r \otimes r') = e_{\xi,\zeta}(r r')
    \end{equation*}
    holds in \( \SymTensor{V}{n,m} \) as well.
    This means that \( e_\xi \otimes e_\zeta \colon R_\xi \otimes R_\zeta \longrightarrow \SymTensor{V}{n} \otimes \SymTensor{V}{m} \hookrightarrow \SymTensor{V}{n,m}\) coincides with \( e_{\xi,\zeta} \colon R_{\xi,\zeta} \longrightarrow \SymTensor{V}{n,m}\) on \( R_\xi \otimes R_\zeta \hookrightarrow R_{\xi,\zeta} \).

    Now suppose that \( x \in |\SymTensor{V}{n} \otimes \SymTensor{V}{m}| \) and that \( y \le x \) in \( \SymTensor{V}{n,m} \).
    Hence there is \( z \in |\SymTensor{V}{n,m}| \) such that \( y + z = x \) in \( \SymTensor{V}{n,m} \).
    We prove that \( y,z \in |\SymTensor{V}{n} \otimes \SymTensor{V}{m}| \).
    Since \( \SymTensor{V}{n} \otimes \SymTensor{V}{m} \hookrightarrow \SymTensor{V}{n,m} \) is sum-reflecting, it suffices to decompose \( y \) and \( z \) as sums (computed in \( \SymTensor{V}{n,m,} \)) of elements in \( \SymTensor{V}{n} \otimes \SymTensor{V}{m} \).
    By Lemma~\ref{lem:module:tensor-peak}, there exist \( u \in |\SymTensor{V}{n}| \) and \( v \in |\SymTensor{V}{m}| \) such that \( x \le u \otimes v \) in \( \SymTensor{V}{n} \otimes \SymTensor{V}{m} \).
    Then
    \begin{align*}
        x &= \sum_{\xi \in \mathcal{M}_n(I), \zeta \in \mathcal{M}_m(I)} e_{\xi,\zeta}(r^{(x)}_{\xi,\zeta}) \\
        y &= \sum_{\xi \in \mathcal{M}_n(I), \zeta \in \mathcal{M}_m(I)} e_{\xi,\zeta}(r^{(y)}_{\xi,\zeta}) \\
        z &= \sum_{\xi \in \mathcal{M}_n(I), \zeta \in \mathcal{M}_m(I)} e_{\xi,\zeta}(r^{(z)}_{\xi,\zeta}) \\
        u \otimes v &= \sum_{\xi \in \mathcal{M}_n(I), \zeta \in \mathcal{M}_m(I)} e_{\xi,\zeta}(r^{(u)}_{\xi} r^{(v)}_{\zeta})
    \end{align*}
    hold in \( \SymTensor{V}{n,m} \),
    where
    \begin{align*}
        r^{(x)}_{\xi,\zeta} &= (\varphi_\xi \otimes \varphi_\zeta)(x) \\
        r^{(y)}_{\xi,\zeta} &= (\varphi_\xi \otimes \varphi_\zeta)(y) \\
        r^{(z)}_{\xi,\zeta} &= (\varphi_\xi \otimes \varphi_\zeta)(z) \\
        r^{(u)}_{\xi} &= \varphi_\xi(u) \\
        r^{(v)}_\zeta &= \varphi_\zeta(v).
    \end{align*}
    Here we use the fact that \( (\varphi_\xi \otimes \varphi_\zeta)(u \otimes v) = \varphi_\xi(u) \varphi_\zeta(v) \).
    Since \( x = y + z \), we have
    \begin{align*}
        r^{(x)}_{\xi,\zeta} 
        &\;=\;
        (\varphi_\xi \otimes \varphi_\zeta)(x)
        \\ &
        \;=\;
        (\varphi_\xi \otimes \varphi_\zeta)(y+z)
        \\ &
        \;\Kle\;
        (\varphi_\xi \otimes \varphi_\zeta)(y) + 
        (\varphi_\xi \otimes \varphi_\zeta)(z)
        \\ &
        \;=\;
        r^{(y)}_{\xi,\zeta} + r^{(z)}_{\xi,\zeta}.
    \end{align*}
    Since \( r^{(u)}_\xi r^{(v)}_\zeta \in |R_\xi \otimes R_\zeta| \) and \( R_\xi \otimes R_\zeta \) is a downward-closed ideal, \( r^{(y)}_{\xi,\zeta} \le r^{(x)}_{\xi,\zeta} \le r^{(u)}_\xi r^{(v)}_\zeta \) implies \( r^{(y)}_{\xi,\zeta}\in |R_\xi \otimes R_\zeta| \) and similarly \( r^{(z)}_{\xi,\zeta}\in |R_\xi \otimes R_\zeta| \).
    Since \( e_\xi \otimes e_\zeta \colon R_\xi \otimes R_\zeta \longrightarrow \SymTensor{V}{n} \otimes \SymTensor{V}{m} \hookrightarrow \SymTensor{V}{n,m}\) coincides with \( e_{\xi,\zeta} \colon R_{\xi,\zeta} \longrightarrow \SymTensor{V}{n,m}\) on \( R_\xi \otimes R_\zeta \hookrightarrow R_{\xi,\zeta} \), we have
    \begin{align*}
        y &= \sum_{\xi \in \mathcal{M}_n(I), \zeta \in \mathcal{M}_m(I)} (e_{\xi} \otimes e_{\zeta})(r^{(y)}_{\xi,\zeta}) \\
        z &= \sum_{\xi \in \mathcal{M}_n(I), \zeta \in \mathcal{M}_m(I)} (e_{\xi} \otimes e_{\zeta})(r^{(z)}_{\xi,\zeta})
    \end{align*}
    where the sum is taken in \( \SymTensor{V}{n,m} \).
    Since \( (e_{\xi} \otimes e_{\zeta})(r^{(y)}_{\xi,\zeta}) \) and \( (e_{\xi} \otimes e_{\zeta})(r^{(y)}_{\xi,\zeta}) \) belong to \( \SymTensor{V}{n} \otimes \SymTensor{V}{m} \), this is a required decomposition.
\end{proof}

Suppose that \( V \) has a countable basis.
By Lemma~\ref{lem:appx:base-of-section}, \( \S^{(n)}V \) and \( \S V \) have orthogonal bases.
Under the orthogonality assumption, if \( V \) has a countable orthogonal dual basis, then so do \( \S^{(n)}V \) and \( \S V \).
To prove the theorem, it suffices to show that \( \S V \) is (isomorphic to) the underlying \( \srig \)-module of \( {!}V \).

By definition, \( \S^{(n)}V \) is a downward-closed sum-reflecting submodule of \( \SymTensor{V}{n} \).
By Lemmas~\ref{lem:appx:submodule} and \ref{lem:appx:tensor-dcsr},
\begin{equation*}
    \S^{(n)}V \otimes \S^{(m)}V \longrightarrow \S^{(n)}V \otimes \SymTensor{V}{m}
\end{equation*}
and
\begin{equation*}
    \S^{(n)}V \otimes \SymTensor{V}{m} \longrightarrow \SymTensor{V}{n} \otimes \SymTensor{V}{m}
\end{equation*}
are downward-closed sum-reflecting submodules.
Hence so is their composite
\begin{equation*}
    \S^{(n)}V \otimes \S^{(m)} V \longrightarrow \SymTensor{V}{n} \otimes \SymTensor{V}{m}.
\end{equation*}

Suppose that \( y \in \S^{(n + m)}V \).
Then there exists \( x \in {!}V \) such that \( y \le \hat{d}^{(n+m)}(x) \) in \( \SymTensor{V}{n+m} \).
Let \( \iota_{n,m} \colon \SymTensor{V}{n+m} \hookrightarrow \SymTensor{V}{n,m} \) be the canonical map coming from the universality of the equaliser \( \SymTensor{V}{n,m} \).
Let \( \chi_{n,m} \colon \SymTensor{V}{n} \otimes \SymTensor{V}{m} \hookrightarrow \SymTensor{V}{n,m} \) be the canonical map, which is downward-closed and sum-reflecting by Lemma~\ref{lem:appx:key-on-comultiplication}.
Then \( \iota_{m,n} \circ \hat{d}^{(n+m)} = \chi_{n,m} \circ (\hat{d}^{(n)} \otimes \hat{d}^{(m)}) \circ \mu \), where \( \mu \colon {!}V \longrightarrow {!}V \otimes {!}V \) is the comultiplication of the comonoid \( {!}V \).
We prove that \( \iota_{m,n}(y) \) belongs to \( \S^{(n)}V \otimes \S^{(m)}V \).
By Lemma~\ref{lem:module:tensor-peak}, there exist \( z, z' \in |{!}V| \) such that \( \mu(x) \le z \otimes z' \).
Hence
\begin{align*}
    \iota_{m,n}(y)
    &\le \iota_{m,n} \hat{d}^{(n+m)}(x) 
    \\
    &= \chi_{n,m} (\hat{d}^{(n)} \otimes \hat{d}^{(m)}) \mu(x)
    \\
    &\le \chi_{n,m}(\hat{d}^{(n)} \otimes \hat{d}^{(m)})(z \otimes z')
    \\
    &= \hat{d}^{(n)}(z) \otimes \hat{d}^{(m)}(z')
\end{align*}
holds in \( \SymTensor{V}{n,m} \).
Since \( \hat{d}^{(n)}(z) \otimes \hat{d}^{(m)}(z') \) belongs to \( \S^{(n)}V \otimes \S^{(m)}V \) and \( \S^{(n)}V \otimes \S^{(m)}V \hookrightarrow \SymTensor{V}{n} \otimes \SymTensor{V}{m} \hookrightarrow \SymTensor{V}{n,m}\) is downward-closed and sum-reflecting (as the composite of downward-closed sum-reflecting monomorphisms), \( \iota_{m,n}(y) \) also belongs to \( \S^{(n)}V \otimes \S^{(m)}V \).
Since \( \S^{(n)}V \otimes \S^{(m)}V \hookrightarrow \SymTensor{V}{n,m} \) is sum-reflecting, there exists a unique monomorphism \( \delta^{n,m} \colon \S^{(n+m)}V \longrightarrow \S^{(n)}V \otimes \S^{(m)}V \) that makes the diagram
\begin{equation*}
    \xymatrix{
        \S^{(n+m)}V \ar[d] \ar[rr]^{\delta^{n,m}} & & \S^{(n)}V \otimes \S^{(m)}V \ar[d] \\
        \SymTensor{V}{n+m} \ar[drr] & & \SymTensor{V}{n} \otimes \SymTensor{V}{m} \ar[d] \\
        & & \SymTensor{V}{n,m}
    }
\end{equation*}
commutative.
Since the outer heptagon and the lower pentagon in the diagram
\begin{equation*}
    \xymatrix{
        {!}V \ar[d]^{\hat{d}^{(n+m)}} \ar[rr]^{\mu} & & {!}V \otimes {!}V \ar[d]^{\hat{d}^{(n)} \otimes \hat{d}^{(m)}} \\ 
        \S^{(n+m)}V \ar[d] \ar[rr]^{\delta^{n,m}} & & \S^{(n)}V \otimes \S^{(m)}V \ar[d] \\
        \SymTensor{V}{n+m} \ar[drr] & & \SymTensor{V}{n} \otimes \SymTensor{V}{m} \ar[d] \\
        & & \SymTensor{V}{n,m}
    }
\end{equation*}
commutes and \( \S^{(n)} V \otimes \S^{(m)}V \hookrightarrow \SymTensor{V}{n} \otimes \SymTensor{V}{m} \hookrightarrow \SymTensor{V}{n,m} \) is a monomorphism, the upper square in the above diagram commutes.

Consider the map
\begin{equation*}
    \bar{\delta} \defe \sum_{n,m < \omega} \delta^{n,m} \circ \pi_{n+m}
    \quad \colon \S V \longrightarrow \CComp{\S V \otimes \S V},
\end{equation*}
where the overline means the completion.
Recall that \( \hat{d} = \sum_n \hat{d}^{(n)} \colon !V \longrightarrow \S V \hookrightarrow \prod_{n < \omega} \SymTensor{V}{n} \).
Since \( \hat{d}^{(n)} = \pi_n \circ \hat{d} \), the above commutative diagram shows that
\begin{equation*}
    \xymatrix{
        {!}V \ar[d]^{\hat{d}} \ar[rr]^{\mu} & & {!}V \otimes {!}V \ar[d]^{\hat{d} \otimes \hat{d}} \\ 
        \S V \ar[rr]^{\delta} & & \CComp{\S V \otimes \S V}
    }
\end{equation*}
commutes.
Since the image of \( \hat{d} \otimes \hat{d} \) is \( \S V \otimes \S V \), we have
\begin{equation*}
    \xymatrix{
        {!}V \ar[d]^{\hat{d}} \ar[rr]^{\mu} & & {!}V \otimes {!}V \ar[d]^{\hat{d} \otimes \hat{d}} \\ 
        \S V \ar[rr]^{\delta} & & \S V \otimes \S V.
    }
\end{equation*}

Then \( (\S V, \delta, \pi_0) \) is a comonoid and \( \hat{d} \) is a comonoid map.
The comonoid \( \S V \) is equipped with \( \pi_0 \colon \S V \longrightarrow V \), and \( \pi_1 \circ \hat{d} = d^{1} \colon !V \longrightarrow V \).
By the universality of \( !V \), we have a comonoid map in the other direction \( f \colon \S V \longrightarrow !V \) such that \( \pi_1 = f \circ d^{1} \).
Then the composite \( \hat{d} \circ f \colon \S V \longrightarrow \S V \) is the identity.
By the universality of \( !V \), \( f \circ \hat{d} \colon !V \longrightarrow !V \) is the identity.
Hence \( \hat{d}^{-1} = f \).

\subsection{Proof of Lemma~\ref{lem:classic:proj-closure}}
    By the same argument as the proof of Lemma~\ref{lem:classic:proj-closure}.

\subsection{Proof of Lemma~\ref{lem:classic:proj-idempotent}}
    Let \( \smod \in \DBMod[\srig] \) and \( \varphi \colon \smod \longrightarrow V \) and \( e \colon V \longrightarrow \smod \) be a retraction to \( V \in \SVec[\srig] \) (i.e.~\( e \circ \varphi = \ident_\smod \)).
    Consider the diagram
    \begin{equation*}
        \xymatrix{
            & V^{\bot\bot\bot\bot} \ar[dl]_{e^{\bot\bot\bot\bot}} \ar[r]^\cong & V^{\bot\bot} \\
            M^{\bot\bot\bot\bot} \ar@{=}[r]
            & M^{\bot\bot\bot\bot} \ar[u]_{\varphi^{\bot\bot\bot\bot}} \ar[r]
            & M^{\bot\bot} \ar[u]_{\varphi^{\bot\bot}}
        }
    \end{equation*}
    where the square commutes by the naturality and the triangle commutes by the functoriarity of \( ({-})^{\bot\bot\bot\bot} \) with \( e \circ \varphi = \ident \).
    Since \( V^{\bot\bot\bot\bot} \longrightarrow V^{\bot\bot} \) is an isomorphism by Lemma~\ref{lem:classic:proj-idempotent}, \( M^{\bot\bot\bot\bot} \longrightarrow M^{\bot\bot} \) is also an isomorphism.

 \section{On the relationship between $\DBMod[{[0,1]}]$ and $\CPCoh$}
\label{sec:appx:pcoh}
Let \( \smod \) be a \( [0,1] \)-module with a countable dual basis \( (e_i, \varphi_i)_{i \in I} \).
Assume \( e_i \neq 0 \) and \( \varphi_i \neq 0 \) for every \( i \).

\begin{lemma}
    Let \( \smod \) be a \( [0,1] \)-module with a dual basis \( (e_i, \varphi_i)_{i \in I} \).
    \( a + b = a + c \) in \( \smod \) implies \( b = c \).
\end{lemma}
\begin{proof}
    If \( a + b = a + c \), then \( \varphi_i(a) + \varphi_i(b) = \varphi_i(a) + \varphi_i(c) \) for every \( i \in I \).
    Because the sum of real numbers are cancellable, we have \( \varphi_i(b) = \varphi_i(c) \) for every \( i \in I \).
    Hence \( b = c \).
\end{proof}

A subset \( J \subseteq I \) of the indexes of a basis \( (e_i, \varphi_i)_{i \in I} \) of \( \smod \) is \emph{closed} if \( j \in J \) and \( \varphi_i(e_j) \neq 0 \) implies \( i \in J \).
That means the canonical representation \( e_j = \sum_i \varphi_i(e_j) \cdot e_i \) is actually a sum only using \( (e_j)_{j \in J} \).
\begin{lemma}
    Let \( \smod \) be a \( [0,1] \)-module with a dual basis \( (e_i, \varphi_i)_{i \in I} \).
    If \( J \subseteq I \) is closed, then so is \( I \setminus J \).
\end{lemma}
\begin{proof}
    Let \( K = I \setminus J \) and \( \ell \in K \).
    Since \( \varphi_\ell \neq 0 \), there exists \( i \) such that \( \varphi_\ell(e_i) \neq 0 \).
    We have
    \begin{align*}
        e_i
        &=
        \sum_{k \in K} \varphi_k(e_{i_0}) \cdot e_k + \sum_{j \in J} \varphi_j(e_i) \cdot e_j.    
    \end{align*}
    Hence, for each \( k' \in K \),
    \begin{align*}
        \varphi_{k'}(e_i)
        &= \sum_{k \in K} \varphi_k(e_i) \varphi_{k'}(e_k) + \sum_{j \in J} \varphi_j(e_i) \varphi_{k'}(e_j)
        \\
        &= \sum_{k \in K} \varphi_k(e_i) \varphi_{k'}(e_k).
    \end{align*}
    Therefore
    \begin{align*}
        &
        \sum_{k \in K} \varphi_k(e_i) \cdot e_k
        \\ = &
        \sum_{k' \in K} \sum_{k \in K} \varphi_k(e_i) \varphi_{k'}(e_k) \cdot e_{k'} + \sum_{j' \in J} \sum_{k \in K} \varphi_k(e_i) \varphi_{j'}(e_k) \cdot e_{j'}
        \\ = &
        \sum_{k' \in K} \varphi_{k'}(e_i) \cdot e_{k'} + \sum_{j' \in J} \sum_{k \in K} \varphi_k(e_i) \varphi_{j'}(e_k) \cdot e_{j'}.
    \end{align*}
    By the previous lemma, \( \sum_{j' \in J} \sum_{k \in K} \varphi_k(e_i) \varphi_{j'}(e_k) \cdot e_{j'} = 0 \), i.e.~\( \varphi_k(e_i) \varphi_{j'}(e_k) = 0 \) for every \( j' \in J \) and \( k \in K \).
    Since \( \varphi_\ell(e_i) \neq 0 \) by the assumption, \( \varphi_{j'}(e_\ell) = 0 \) for every \( j' \in J \).
\end{proof}

\begin{lemma}
    Let \( \approx \) be a relation on \( I \) defined by \( i \approx j \) iff \( \varphi_j(e_i) \neq 0 \).
    Then \( \approx \) is an equivalence relation.
\end{lemma}
\begin{proof}
    It is easy to see that \( \approx \) is transitive.
    We prove that \( \approx \) is symmetric.
    Then \( \approx \) is a partial equivalence relation; it is an equivalence relation since \( e_i \neq 0 \), i.e.~\( \forall i. \exists j. i \approx j \).

    \newcommand{\support}{\mathord{\mathrm{supp}}}
    Let \( \support(i) := \{ j \in I \mid \varphi_j(e_i) \neq 0 \} \).
    Then \( i \approx j \) if and only if \( j \in \support(i) \).
    Note that \( \support(i) \) is closed for each \( i \in I \); actually it is the minimum closed subset containing \( i \).

    Assume that \( j \in \support(i) \).
    By the previous lemma, \( I \setminus \support(j) \) is also closed.
    Since the intersection of two closed set is closed, \( \support(i) \cap (I \setminus \support(j) = \support(i) \setminus \support(j) \) is closed.
    Since \( \emptyset \subsetneq \support(j) \subseteq \support(i) \), \( \support(i) \cap \support(j) \) is a proper subset of \( \support(i) \).
    If \( i \notin \support(j) \), then \( i \in (\support(i) \setminus \support(j) \), which contradicts the fact that \( \support(i) \) is the minimum closed set containing \(i\).
\end{proof}

Let us choose a set \( J \subseteq I \) of representatives of each equivalence class of \( \approx \).
\begin{lemma}
    Every element \( x \in |\smod| \) can be uniquely written as a sum \( x = \sum_{j \in J} r_j \cdot e_j \) of \( (e_j)_{j \in J} \) with positive real coefficients \( r_j \in [0,\infty) \).
\end{lemma}
\begin{proof}
    We first prove the uniqueness of the representation.
    Let \( (s_{j,k})_{k < \omega} \) be a family of \( [0,1] \) such that \( r_j = \sum_k s_{j,k} \).
    Then \( x = \sum_{j \in J} \sum_{k < \omega} s_{j,k} \cdot e_j \).
    Let \( i \in I \) and \( \ell \in J \) be the unique index such that \( \varphi_i(e_\ell) \neq 0 \).
    Then
    \begin{align*}
        \varphi_i(x)
        &= \sum_{j \in J} \sum_{k < \omega} s_{j,k} \cdot \varphi_i(e_j)
        \\
        &= \sum_{k < \omega} s_{\ell,k} \varphi_i(e_\ell)
        \\
        &= r_\ell \varphi_i(e_\ell).
    \end{align*}
    Since \( (\varphi_i)_{i \in I} \) is jointly monic, the representation is unique.
    
    We show the existence.
    Let \( f,g \colon \smod \longrightarrow \smod \) be a \( [0,1] \)-linear map defined by
    \( f(x) = \sum_{i \notin J} \varphi_i(x) \cdot e_i \)
    and \( g(x) = \sum_{j \in J} \varphi_j(x) \cdot e_j \).
    They are well-defined since they are partial sums of \( \sum_{i \notin I} \varphi_i(x) \cdot e_i = x \).
    Then \( x = f(x) + g(x) \) and thus \( x = f^n(x) + \sum_{k < n} g(f^n(x)) \) for each \( n \).
    We show that \( \lim_{n \to \infty} \varphi_i(f^n(x)) = 0 \) and thus \( x = \sum_{n < \omega} g(f^n(x)) \), i.e.~\( x = \sum_{j \in J} \sum_{n < \omega} \varphi_j(f^n(x)) \cdot e_j \).
    Then \( \sum_{n < \omega} \varphi_j(f^n(x)) \) converges in \( \mathbb{R} \); otherwise \( \varphi_j(x) = \sum_{n < \omega} \varphi_j(f^n(x)) \varphi_j(e_j) \) does not converge in \( [0,1] \). 

    Let \( j \in J \).
    Then \( \varphi_j(x) = \varphi_j(f(x)) + \varphi_j(g(x)) = \varphi_j(f(x)) + \sum_{j' \in J} \varphi_{j'}(x) \varphi_j(e_{j'}) = \varphi_j(f(x)) + \varphi_j(x) \varphi_j(e_j) \).
    Hence \( \varphi_j(f(x)) = (1 - \varphi_j(e_j)) \varphi_j(x) \) and thus \( \varphi_j(f^n(x)) = (1- \varphi_j(e_j))^n \varphi_j(x) \).
    Since \( j \approx j \), \( \varphi_j(e_j) \neq 0 \) and thus \( (1-\varphi_j(e_j))^n \longrightarrow 0 \) (\( n \longrightarrow \infty \)).

    For each \( i \in I \), let \( j \in J \) be the representative, i.e.~\( i \approx j \).
    Then \( \varphi_j(x) = \varphi_j(\sum_{i' \approx j} \varphi_{i'}(x) \cdot e_{i'}) = \sum_{i' \approx j} \varphi_{i'}(x) \varphi_{j}(e_{i'}) \ge \varphi_{i}(x) \varphi_j(e_{i}) \).
    Since \( \varphi_j(e_i) > 0 \) and \( \varphi_j(f^n(x)) \longrightarrow 0 \) (\( n \longrightarrow \infty \)), we conclude that \( \varphi_i(f^n(x)) \longrightarrow 0 \) (\( n \longrightarrow \infty \)).
\end{proof}

\begin{theorem}
    A \( [0,1] \)-module \( \smod \) with a countable dual basis is isomorphic to a downward-closed full-subalgebra \( X \hookrightarrow \prod_{i \in I} \mathbb{R}_{\ge 0} \) for some countable set \( I \) such that \( \pi_i(X) = \{ \pi_i(x) \mid x \in X \} \) is bounded for every \( i \in I \).
\end{theorem}

Let us briefly discuss the relationship between the exponential modality in \cite{Danos2011} and our modality, both of which are of the form \( (\S\smod)^{\bot\bot} \) for some \( \S \).
A difference can be found in the definition of \( \S\smod \).
In the former,
\begin{equation*}
    \S\smod = \{ !x \in |\prod_{n} \SymTensor{\smod}{n}| \mid x \in |\smod| \}, 
\end{equation*}
where \( !x = \sum_n \overbrace{x \otimes \dots \otimes x}^n \), whereas in the latter,
\begin{equation*}
    \S'\smod = \{ y \in |\prod_{n} \SymTensor{\smod}{n}| \mid x \in |!\smod| \}. 
\end{equation*}
Since \( !x \) can be seen as an element of \( !\smod \) for every \( x \in |\smod| \), the latter is lager.
Hence there is a natural \( (\S \smod)^{\bot\bot} \longrightarrow (\S' \smod)^{\bot\bot} \), but the other direction is unclear.

The point here is that the exponential modality \( (\S\smod)^{\bot\bot} \) of \cite{Danos2011} is actually cofree~\cite{Abou-Saleh2013} in \( \CPCoh \).
Since \( (\S' \smod)^{\bot\bot} \) has an induced comonoid structure, there exists a canonical comonoid map \( (\S' \smod)^{\bot\bot} \longrightarrow (\S \smod)^{\bot\bot} \).
Then \( (\S' \smod)^{\bot\bot} \longrightarrow (\S \smod)^{\bot\bot} \longrightarrow (\S \smod)^{\bot\bot} \longrightarrow (\S' \smod)^{\bot\bot} \) is identity, and hence \( (\S' \smod)^{\bot\bot} \) and \( (\S \smod)^{\bot\bot} \) are isomorphic by the cofreeness of \( (\S \smod)^{\bot\bot} \).

\end{document}